\documentclass[onecolumn,11pt,draftcls]{IEEEtran}
\usepackage{amsbsy}%,color,multicol}
\usepackage{floatflt} % text flows around figures

\usepackage{amsmath}
\usepackage{amssymb}
\usepackage{times}
\usepackage{graphicx}
\usepackage{xspace}
\usepackage{paralist} % compact and in-paragraph* lists
\usepackage{setspace} % buggy somehow
\usepackage{xypic}
\xyoption{curve}
\usepackage{latexsym}
\usepackage{theorem}
\usepackage{ifthen}
\usepackage{stackrel}
\usepackage{subfigure}
\usepackage{booktabs}
\newcommand{\mypar}[1]{\vspace{0.03in}\noindent{\bf #1.}}

\newcounter{brojac}

{\theoremheaderfont{\it} \theorembodyfont{\rmfamily}
\newtheorem{assumption}[brojac]{Assumption}
}

{\theoremheaderfont{\it} \theorembodyfont{\rmfamily}
\newtheorem{theorem}{Theorem}
\newtheorem{lemma}[theorem]{Lemma}

\newtheorem{assumptions}[theorem]{Assumptions}

\newtheorem{corollary}[theorem]{Corollary}
}

\newcommand{\E}{\mathbb E}

\title{Large Deviations Performance of Consensus+Innovations Distributed Detection with Non-Gaussian Observations}
\author{Dragana Bajovi\'c, Du$\breve{\mbox{s}}$an Jakoveti\'c, Jos\'e M.~F.~Moura, Jo\~ao
Xavier, and Bruno Sinopoli% <-this
\thanks{The work of the first, second, fourth, and fifth authors is partially supported by grants CMU-PT/SIA/0026/2009,
SFRH/BD/33517/2008 (through the Carnegie Mellon/Portugal Program
managed by ICTI) and by grant PTDC/EEA-CRO/104243/2008 from Funda\c{c}\~{a}o para a Ci\^encia e Tecnologia and
also by
ISR/IST plurianual funding (POSC program, FEDER). The work of the third author is partially supported by NSF under grants CCF-1011903 and CCF-1018509, and by AFOSR grant
FA95501010291. Dragana Bajovi\'{c} and Du$\breve{\mbox{s}}$an Jakoveti\'c hold fellowships
from the Carnegie Mellon/Portugal Program.}% <-this % stops a space
\thanks{Dragana Bajovi\'{c} and Du$\breve{\mbox{s}}$an Jakoveti\'c are with the
Institute for Systems and Robotics
(ISR), Instituto Superior T\'{e}cnico (IST), Technical University of Lisbon, Lisbon, Portugal, and with
the Department of Electrical and Computer Engineering, Carnegie Mellon
University, Pittsburgh, PA, USA, {\tt\small dbajovic@andrew.cmu.edu,
djakovet@andrew.cmu.edu}}%
\thanks{Jos\'e M.~F.~Moura and Bruno Sinopoli are with the Department of
Electrical and Computer
Engineering, Carnegie Mellon University, Pittsburgh, 15213 PA, USA, {\tt\small
moura@ece.cmu.edu, brunos@ece.cmu.edu}}%
\thanks{Jo\~ao Xavier is with the Institute for Systems and Robotics (ISR),
Instituto Superior T\'{e}cnico (IST), Technical University of Lisbon, Lisbon, Portugal, {\tt\small
jxavier@isr.ist.utl.pt}}
}%
\begin{document}

\maketitle
\vspace{-9mm}
\begin{abstract}\small
We establish the large deviations asymptotic performance (error exponent) of consensus+innovations distributed detection
 over random networks with \emph{generic} (non-Gaussian)
  sensor observations. At each time instant, sensors 1)~combine theirs with the decision variables of their neighbors (consensus)
  and 2)~assimilate their new observations (innovations).
    This paper shows for general non-Gaussian distributions that consensus+innovations distributed detection exhibits a
   phase transition behavior with respect to the network degree of connectivity. Above a threshold,
      distributed is as good as centralized, with the same optimal asymptotic detection performance, but, below
     the threshold,  distributed detection is suboptimal with respect to
      centralized detection.  We determine this threshold and quantify the performance loss below threshold.
      Finally, we show the dependence of the threshold and performance on the distribution
       of the observations: distributed detectors over the same random network, but with
       different observations' distributions, for example, Gaussian,  Laplace, or quantized, may have different asymptotic performance, even when the corresponding centralized detectors have the same asymptotic performance.
\end{abstract}
\hspace{-1.43cm}\textbf{Keywords:} {\small Consensus+innovations, performance analysis, Chernoff information, non-Gaussian distributions, distributed detection,
random network, information flow, large deviations.}
\newpage
\section{Introduction}
\renewcommand{\baselinestretch}{1.2}
\label{Intro}
Consider a distributed detection scenario where $N$ sensors are connected by a generic network with intermittently failing links.  The sensors perform consensus+innovations distributed detection; in other words,  at each time $k$, each sensor $i$ updates its local decision variable $x_i(k)$ by: 1)~sensing and processing a new measurement to create an intermediate variable; and
2)~weight averaging it with its neighbors' intermediate decision variables. We showed in~\cite{GaussianDD} that, when the sensor observations are Gaussian, the consensus+innovations distributed detector exhibits a phase transition. When the network connectivity  %  measured by $|\log r|$, where $r:=\lambda_2 \left( \mathbb{E} \left[ W^2(k) \right]  \right)$ is the second largest eigenvalue of the second moment of the (random) consensus matrix $W(k)$. That is, when
%  $|\log r|$
 is above a threshold, then the distributed detector is asymptotically optimal, i.e., asymptotically equivalent to the optimal centralized detector that collects the observations of all sensors.

  This paper establishes the asymptotic performance of distributed detection over random networks for \emph{generic, non-Gaussian} sensor observations. We adopt as asymptotic performance measure the
   exponential decay rate of the Bayes error probability (error exponent). We show that phase transition behavior emerges with non-Gaussian observations and demonstrate how the optimality threshold is a function of the log-moment generating function of the sensors' observations and of the number of sensors $N$. This reveals %for the performance of distributed detection
   %
   %\begin{equation}
%   \label{eqn_threshold}
%    |\log r|\geq \max\left\{\Lambda_0(N \lambda^\star)-N\Lambda_0(\lambda^\star), \Lambda_0(1-N(1- \lambda^\star))-N\Lambda_0(\lambda^\star)\right\},
%   \end{equation}
   %
%where $\lambda^\star = \mathrm{arg\,max}_{\lambda \in [0,1]} \left\{ -\Lambda_0(\lambda)\right\}$.
 a very interesting interplay  between the distribution of the sensor observations (e.g., Gaussian or Laplace) and the rate of diffusion (or connectivity) of the network %We measure the connectivity by a parameter, $|\log r|$, where $r$ is given in~\eqref{eqn-r-def}.
% Consider a  fixed network (with given $|\log r|$),
% and two types of distributions, say, Gaussian and Laplace,  with \emph{equal per sensor Chernoff information} $C_{\mathrm{ind},\mathrm{G}}=
% C_{\mathrm{ind},\mathrm{L}} = :C_{\mathrm{ind}}$. Since the per-sensor Chernoff information determines the error exponent of centralized detection, which equals $NC_{\mathrm{ind}}$, the two corresponding centralized detectors have the same error exponent. However, the two corresponding
% distributed detectors (Gaussian and Laplace) are not equivalent and may have different error exponents. In fact, as
(measured by a parameter $|\log r| \in [0,\infty)$ defined in Section~\ref{sec-Problem-formulation}): for a network with the same connectivity, a distributed detector with say, Laplace observations distributions, may match the optimal asymptotic performance of the centralized detector, while the distributed detector for Gaussian observations may be suboptimal, even though the centralized detectors for the two distributions, Laplace and Gaussian, have the same  optimal asymptotic performance.

For distributed detection, we determine the range on the detection threshold $\gamma$ for which
 each sensor achieves exponentially fast decay of the error
 probability (strictly positive error exponent), and
  we find the optimal $\gamma$ that maximizes the error exponent.
  Interestingly, above the critical (phase transition) value for
  the network connectivity $|\log r|$, the optimal detector threshold is $\gamma=0$,
  mimicking the (asymptotically) optimal threshold for the centralized detector.
   However, below the critical connectivity, we show by a
   numerical example that the optimal distributed detector threshold might be non zero.

\mypar{Brief review of the literature}
Distributed detection has been extensively studied, in the context of parallel fusion architectures, e.g., \cite{Moura-saddle-point,Varshney-I,Poor-II,Veraavali-SPMag,Veraavali,Veraavali-LDP,sensor-selection},
 consensus-based detection~\cite{Moura-detection-consensus,consensus-detection,aldosarimouramay06,MouraInference},
  and, more recently, consensus+innovations distributed inference,
%  ~\cite{allerton,DusanNoise,running-consensus-detection,Sayed-detection,Sayed-detection-2,cattivellisayed-2011,Soummya-Detection-Noise}.
%This paper belongs to the recently growing literature on consensus+innovations
%type distributed inference, see, e.g.,
see, e.g., \cite{Sayed-LMS,SoummyaEst,Giannakis-LMS,Giannakis-LMS-2,Sayed-LMS-new} for distributed estimation, and \cite{Sayed-detection,Sayed-detection-2,running-consensus-detection,cattivellisayed-2011,Soummya-Detection-Noise,stankovic-change-detection,stankovic-conference} for distributed detection.
Different variants of consensus+innovations distributed detection algorithms have been proposed; we analyze here
  running consensus, the variant in~\cite{running-consensus-detection}.

 Reference~\cite{running-consensus-detection} considers asymptotic optimality of running consensus, but in a framework that is very different from ours. Reference~\cite{running-consensus-detection} studies the asymptotic performance of the distributed detector where the means of the sensor observations under the two hypothesis become closer and closer (vanishing signal to noise ratio (SNR)), at the rate of $1/\sqrt k$, where~$k$ is the number of observations. For this problem, there is an asymptotic, non-zero, probability of miss and an asymptotic, non-zero, probability of false alarm. Under these conditions, running consensus is as efficient as the optimal centralized detector, \cite{Kassam}, as long as the network is connected on average. Here, we assume that the means of the distributions stay fixed as~$k$ grows. We establish, through large deviations,   the rate (error exponent) at which the error probability decays to zero as~$k$ goes to infinity. We show that connectedness on average is not sufficient for running consensus to achieve the optimality of centralized detection; rather, phase change occurs, with distributed becoming as good as centralized, when the network connectivity, measured by $|\log r|$, exceeds a certain threshold.

 We distinguish this paper from our prior work on the performance analysis of running consensus. In~\cite{allerton}, we studied \emph{deterministically} time varying networks and \emph{Gaussian} observations, and in~\cite{DusanNoise},
 we considered a different consensus+innovations detector with Gaussian observations and additive communication noise. Here, we consider \emph{random} networks, \emph{non-Gaussian} observations, and noiseless communications. Reference~\cite{GaussianDD} considers \emph{random} networks and \emph{Gaussian}, spatially correlated observations. In contrast, here the observations are \emph{non-Gaussian} spatially independent. We proved our results in~\cite{GaussianDD} by using the \emph{quadratic} nature of the Gaussian log-moment generating function. For general \textit{non-Gaussian} observations, the log-moment generating function is no longer quadratic, and the arguments in~\cite{GaussianDD} no longer apply; we develop a more general methodology that establishes the optimality threshold in terms of the log-moment generating function of the log-likelihood ratio.  We derive our results from \emph{generic} properties of the log-moment generating function like \emph{convexity} and \emph{zero value at the origin}. Finally, while reference~\cite{GaussianDD} and our
 other prior work considered zero detection threshold $\gamma=0$, here we
 extend the results for generic detection thresholds $\gamma$. Our analysis reveals
 that, when $|\log r|$ is above its critical value, the zero detector threshold $\gamma=0$
  is (asymptotically) optimal. When $|\log r|$ is below the critical value,
  we compute the best detector threshold $\gamma = \gamma^\star$, which may be non-zero in general.

 Our analysis shows the impact of the distribution of the sensor observations on the performance of distributed detection:
  distributed detectors (with different distributions of the  sensors observations) can have different asymptotic performance, even though the corresponding centralized detectors are equivalent, as we will illustrate in detail in Section~\ref{sec-Examples}.

\mypar{Paper outline} Section~\ref{sec-Problem-formulation} introduces the network and sensor observations models and presents the consensus+innovations distributed detector.
%We describe the network by a sequence
% of independent identically distributed (i.i.d.)
% (symmetric, stochastic) consensus matrices $\{W(k)\}$.
% We assume both spatial and temporal independence, but
% generic (non-Gaussian) distributions of the sensor observations.  We require that the log-moment generating function of the log-likelihood ratio of a sensor's observation be finite for any real argument. This is satisfied
%    in a wide range of problems, e.g., when: 1) the sensors observations are of the type
%     signal+noise under hypothesis $H_1$, and noise only under $H_0$,
%     where the noise has the same generic distribution under both hypothesis and satisfies a mild technical condition (detailed in
%     Section~\ref{sec-Problem-formulation}); or 2) the sensors' observations take values in a discrete, finite alphabet, the case
%     relevant when the sensors digitize the observations.
Section~\ref{Main-result} presents and proves our main results on the asymptotic performance of the distributed detector. For a cleaner exposition, this section proves the results for (spatially) identically distributed sensor observations. Section~\ref{sec-Examples} illustrates our results on several types of
sensor observation distributions, namely, Gaussian, Laplace, and discrete valued distributions,
discussing the impact of these distributions on distributed detection performance. Section~\ref{section-extensions} extends our main results to non-identically distributed sensors' observations. Finally, Section~\ref{section-conclusion}  concludes the paper.

\mypar{Notation}
We denote by: $A_{ij}$ the $(i,j)$-th entry of a matrix $A$; $a_i$ the $i$-th entry of a vector $a$; $I$, $1$, and $e_i$, respectively, the identity matrix, the column vector with unit entries, and the $i$-th column of $I$; $J$ the $N \times N$ ideal consensus matrix $J:=(1/N)11^\top$; $\| \cdot \|_l$ the vector (respectively, matrix) $l$-norm of its vector (respectively, matrix) argument; $\|\cdot\|=\|\cdot\|_2$ the Euclidean (respectively, spectral) norm of its vector
(respectively, matrix) argument; $\mu_i(\cdot)$ the $i$-th largest eigenvalue; $\mathbb E \left[ \cdot \right]$ and $\mathbb P \left( \cdot\right)$ the expected value and probability operators, respectively; $\mathcal{I}_{\mathcal{A}}$ the indicator function of the event $\mathcal{A}$; $\nu^N$ the product measure of $N$ i.i.d. observations drawn from the distribution with measure $\nu$; $h^\prime(z)$ and $h^{\prime \prime}(z)$ the first and the second derivatives of the function $h$ at point $z$.

\section{Problem formulation}
\label{sec-Problem-formulation}
This section introduces the sensor observations model,
 reviews the optimal centralized detector, and presents
 the consensus+innovations distributed detector. The
 section also reviews
 relevant properties of the log-moment generating function of a sensor's log-likelihood ratio that are needed in the sequel.
\subsection{Sensor observations model}
\label{subsec-Problem-model}
We study the binary hypothesis testing problem $H_1$ versus $H_0$. We consider
 a network of $N$ nodes where $Y_i(t)$ is the observation of sensor $i$ at time $t$, where $i=1,\ldots,N$, $t=1,2,\ldots$

\begin{assumption}
\label{ass_1}
The sensors' observations $\left\{ Y_i(t) \right\}$ are independent and identically distributed (i.i.d.) both in time and in space, with distribution $\nu_1$ under hypothesis $H_1$ and $\nu_0$ under $H_0$:
\begin{equation}
Y_i(t)\,\sim\,\left\{ \begin{array}{lr}   \nu_1, \;\;\;\; H_1\\  \nu_0, \;\;\;\; H_0 \end{array} \right.\:\:, i=1,\ldots,N,\: t=1,2,\ldots
\end{equation}
The distributions $\nu_1$ and $\nu_0$ are mutually absolutely continuous, distinguishable measures. The prior probabilities
$\pi_1 = \mathbb{P}(H_1)$ and $\pi_0 = \mathbb{P}(H_0)=1-\pi_1$ are in $(0,1)$.
\end{assumption}
By spatial independence, the joint distribution of the observations of all sensors
\begin{equation}
\label{eqn-Y-t}
Y(t):=\left(Y_1(t),\ldots, Y_N(t)\right)^\top
 \end{equation}
 at any time $t$ is $\nu_1^N$ under $H_1$ and $\nu_0^N$ under $H_0$. Our main results in Section~{III}
 are derived under Assumption~\ref{ass_1}. Section~V extends them to non-identical (but still independent) sensors' observations.
\subsection{Centralized detection, log-moment generating function (LMGF), and optimal error exponent}
\label{subsec-Centralized-detection}
%
%When the measurements $Y_i(t)$ are absolutely continuous random variables,
%they are described by the density function $p_l(\cdot)$, $l=0,1$.
The log-likelihood ratio of sensor $i$ at time $t$ is $L_i(t)$ and given by
\[
L_i(t)=\log \frac{f_1\left( Y_i(t)\right)} {f_0\left( Y_i(t)\right)},
\]
where, $f_l(\cdot)$, $l=0,1,$ is 1) the probability density function corresponding to $\nu_l$, when $Y_i(t)$ is an absolutely continuous random
variable; or 2) the probability mass function corresponding to $\nu_l$,
when $Y_i(t)$ is discrete valued.

% where $\frac{d \nu_1}{d \nu_0}(\cdot)$ is the Radon-Nikodym derivative
%of $\nu_1$ with respect to $\nu_0$.
%The log-likelihood ratio at time $t$ for the vector $Y(t)$ in \eqref{eqn-Y-t} of the observations of all sensors is:
%\begin{equation}
%\label{eqn-sum-L-i}
%\sum_{i=1}^N L_i(t).
%\end{equation}
%
Under Assumption~\ref{ass_1}, the log-likelihood ratio test for $k$ time observations from all sensors, for a threshold $\gamma$~is:
\footnote{In~\eqref{eq-centralized-dec-rule}, we re-scale
 the spatio-temporal sum of the log-likelihood ratios $L_i(t)$ by dividing
 the sum by $Nk$. Note that we can do so without loss of generality,
 as the alternative test without re-scaling is: $\sum_{t=1}^k \sum_{i=1}^N   L_i(t)  \stackrel[H_0]{H_1}{\gtrless}\gamma^\prime,$
  with $\gamma^\prime=N k \gamma.$}
\begin{equation}\label{eq-centralized-dec-rule}
D(k):=\frac{1}{Nk}\sum_{t=1}^k \sum_{i=1}^N   L_i(t)  \stackrel[H_0]{H_1}{\gtrless}\gamma.
\end{equation}

\mypar{Log-moment generating function (LMGF)}
We introduce the LMGF of $L_i(t)$ and its properties that
play a major role in assessing the performance of distributed detection.

Let $\Lambda_l$ ($l=0,1$) denote the LMGF for the log-likelihood ratio under hypothesis $H_l$:
\begin{equation}
\label{eq-log-momentgen-fcn}
\Lambda_l: \mathbb{R} \longrightarrow \left(-\infty,+\infty\right],\;\;\; \Lambda_l(\lambda)=\log \E \left[e^{\lambda \, L_1(1)}\,| \,H_l\right].
%\Lambda_l: \mathbb{R} \longrightarrow \left(-\infty,+\infty\right],\;\;\; \Lambda_l(\lambda)=\log \E \left[e^{(-1)^{l}\,\lambda \, L_1(1)}\,| \,H_l\right].
\end{equation}
In \eqref{eq-log-momentgen-fcn}, $L_1(1)$ replaces $L_i(t)$, for arbitrary $i=1,...,N$, and $t=1,2,...$, due
to the spatial and temporal identically distributed observations, see Assumption~\ref{ass_1}.
%
%Next lemma states the properties of $\Lambda_0(\cdot)$ that will be of interest in the subsequent analysis.
\begin{lemma}
\label{lemma-log-momentgen-fcns}
 Consider Assumption~\ref{ass_1}. For $\Lambda_0$ and $\Lambda_1$ in~\eqref{eq-log-momentgen-fcn} the following holds:
\begin{enumerate}[(a)]
\label{lemma-convexity-lmgf}
\item $\Lambda_0$ is convex;
\label{lemma-log-momentgen-fcns-part-a}
\item $\Lambda_0(\lambda)\in\left(-\infty,0\right)$, for $\lambda\in(0,1)$, $\Lambda_0(0)=\Lambda_0(1)=0$, 
and $\Lambda_l^\prime(0)=\mathbb{E} \left[ L_1(1) | H_l\right]$, $l=0,1$;
\label{lemma-log-momentgen-fcns-part-b}
\item
\label{lemma-log-momentgen-fcns-part-c}
$\Lambda_1(\lambda)$ satisfies:
\begin{equation}
\label{eq-counterpart-for-H1}
\Lambda_1(\lambda)=\Lambda_0(\lambda+1),\;\;\;\mathrm{for}\;\;\lambda \in\mathbb{R}.
\end{equation}
\end{enumerate}
\end{lemma}
\begin{proof} For a proof of~\eqref{lemma-log-momentgen-fcns-part-a} and~\eqref{lemma-log-momentgen-fcns-part-b}, see~\cite{Hollander}. Part~\eqref{lemma-log-momentgen-fcns-part-c} follows from the definitions of $\Lambda_0$~and~$\Lambda_1$, which we show here for the case when the distributions $\nu_1$ and $\nu_0$ are absolutely continuous (the proof for discrete distributions is similar):
\begin{align*}
\Lambda_1(\lambda)=\log \E \left[e^{\lambda L_1(1)}| H_1\right]&=\log\int_{y \in \mathbb R} \left(\frac{f_1(y)}
{f_0(y)}\right)^{\lambda} f_1(y) d y\\
&=\log\int_{y \in \mathbb R} \left(\frac{f_1(y)}{f_0(y)}\right)^{1+\lambda} f_0(y) d y=\Lambda_0(1+\lambda).
\end{align*}
\vspace{-3mm}
%
%\begin{eqnarray*}
%\Lambda_1(\lambda)=\log \E \left[e^{\lambda L_1(1)}| H_1\right]&=&\log\int_{y \in \mathbb R} \left(\frac{d \nu_1}{d \nu_0}(y)\right)^{\lambda} d\nu_1(y)\\
%&=&\log\int_{y \in \mathbb R} \left(\frac{d \nu_1}{d \nu_0}(y)\right)^{1-\lambda} d\nu_0(y)\\
%&=&\Lambda_0(1-\lambda).
%\end{eqnarray*}
%
%
\end{proof}

We further assume that the LMGF of a sensor's observation is finite.
\begin{assumption}
\label{ass_finite_LMGF}
$\Lambda_0(\lambda)<+\infty$, $\forall \lambda \in {\mathbb R}$.
\end{assumption}
In the next two remarks, we
give two classes of problems when Assumption~\ref{ass_finite_LMGF} holds.

\mypar{Remark~{I}} We consider the  signal+noise model:
\begin{equation}
\label{eqn-noise+signal}
Y_i(t)\,=\,\left\{ \begin{array}{ll}   m+n_i(k), & H_1\\  n_i(k), & H_0. \end{array} \right.
\end{equation}
 Here $m \neq 0$ is a constant signal and $n_i(k)$ is a zero-mean additive noise with density function $f_n(\cdot)$ supported on $\mathbb R$; we rewrite
 $f_n(\cdot)$, without loss of generality, as $f_n(y) = c\,e^{-g(y)}$, where $c>0$ is a constant. Then,
 the Appendix shows that Assumption~\ref{ass_finite_LMGF} holds under the following mild technical condition: either one of
 \eqref{eqn-polynomial} or~\eqref{eqn-logarithmic} \textit{and} one of \eqref{eqn-polynomial-2} or~\eqref{eqn-logarithmic-2} hold:
 \begin{eqnarray}
 \label{eqn-polynomial}
 \lim_{y \rightarrow +\infty} \frac{g(y)}{|y|^{\tau_{+}}} &=& \rho_{+},\:\:\mathrm{for\,some\,\,}\rho_{+},\tau_{+} \in (0,+\infty)\\
  \label{eqn-logarithmic}
 \lim_{y \rightarrow +\infty} \frac{g(y)}{\left(\log (|y|)\right)^{\mu_{+}}} &=& \rho_{+},\:\:
 \mathrm{for\,some\,\,}\rho_{+} \in (0,+\infty),\:\mu_{+}\in(1,+\infty)\\
% \end{eqnarray}
% and, either one of the following two equations,~, holds:
% \begin{eqnarray}
 \label{eqn-polynomial-2}
  \lim_{y \rightarrow -\infty} \frac{g(y)}{|y|^{\tau_{-}}} &=& \rho_{-},\:\:\
  \mathrm{for\,some\,\,}\rho_{-},\tau_{-} \in (-\infty,0)\\
% \end{eqnarray}
% or there exist $\rho_{+}, \rho_{-} \in (0,+\infty)$, $\mu_{+},\mu_{-} \in (1,\infty)$, such that:
%  \begin{eqnarray}
%  \label{eqn-logarithmic}
% \lim_{y \rightarrow +\infty} \frac{g(y)}{\left(\log (|y|)\right)^{\mu_{+}}} &=& \rho_{+}\\
\label{eqn-logarithmic-2}
  \lim_{y \rightarrow -\infty} \frac{g(y)}{\left( \log(|y|)\right)^{\mu_{-}}} &=& \rho_{-},
   \:\: \mathrm{for\,some\,\,}\rho_{-} \in (0,-\infty),\:\mu_{-}\in(1,+\infty).
 \end{eqnarray}
In~\eqref{eqn-logarithmic} and~\eqref{eqn-logarithmic-2}, we can also allow either (or both) $ \mu_{+}, \mu_{-}$ to equal 1,
but then the corresponding $\rho$ is in $(1,\infty)$. Note that $f_n(\cdot)$ need not be symmetric, i.e., $f_n(y)$ need not be equal to $f_n(-y)$. Intuitively, the tail of the density $f_n(\cdot)$ behaves regularly, and $g(y)$ grows either
like a polynomial of arbitrary finite order in $y$, or slower, like a power $y^{\tau}$, $\tau \in (0,1)$,
or like a logarithm $c (\log y)^\mu$. The class of admissible densities $f_n(\cdot)$ includes, e.g.,
power laws $c y^{-p}$, $p>1$, or the exponential families $e^{\theta \,\phi(y)-A(\theta)}$, $A(\theta): = \log \int_{y=-\infty}^{+\infty}e^{\theta \phi(y)} \chi(d y)$, with: 1) the Lebesgue base measure $\chi$;
2) the polynomial, power, or logarithmic potentials $\phi(\cdot)$;
 and 3) the canonical set of parameters $\theta \in \Theta  =
 \left\{\theta:\,\, A(\theta)<+\infty \right\}$,~\cite{weinwright}.
 %\footnote{We do not need $n_i(k)$ to be zero mean,
 %and so }

\mypar{Remark~{II}} Assumption~\ref{ass_finite_LMGF}
 is satisfied if $Y_i(k)$ has arbitrary (different) distributions under $H_1$ and $H_0$ with the same, compact support;
  a special case is when $Y_i(k)$ is discrete, supported on a finite alphabet.

\mypar{Centralized detection: Asymptotic performance}
We consider briefly the performance of the centralized detector that will benchmark the performance of the distributed detector.
Denote by $\gamma_l : = \mathbb{E}\left[ L_1(1)|H_l\right]$, $l=0,1.$ It can be shown~\cite{DemboZeitouni} that
 $\gamma_0<0$ and $\gamma_1>0$. Now, consider the centralized detector in~\eqref{eq-centralized-dec-rule}
  with constant thresholds $\gamma$, for all $k$, and denote
  by:
  \begin{equation}
  \label{eqn-centralized-alpha-beta-pe}
  \alpha(k,\gamma) = \mathbb{P}\left( D(k)\geq \gamma|H_0\right),
  \:\beta(k,\gamma) = \mathbb{P}\left( D(k)< \gamma|H_1\right),
  : P_{\mathrm{e}}(k,\gamma) = \alpha(k,\gamma) \pi_0 + \beta(k,\gamma) \pi_1,
  \end{equation}
respectively, the probability of false alarm, probability of miss, and
Bayes (average) error probability. In this paper,
we adopt the \emph{minimum Bayes error probability} criterion, both for
the centralized and later for our distributed detector, and, from now on,
we refer to it simply as the error probability.
A standard Theorem (Theorem 3.4.3.,~\cite{DemboZeitouni})
says that, for any choice of $\gamma \in (\gamma_0,\gamma_1)$,
the error probability decays exponentially fast to zero in $k$.
For $\gamma \notin (\gamma_0,\gamma_1)$, the error
probability does not converge to zero at all. To see
this, assume that $H_1$ is true, and let $\gamma \geq \gamma_1$.
Then, by noting that $\mathbb{E}[D(k)|H_1] = \gamma_1$, for all $k$, we have
that $\beta(k,\gamma) = \mathbb{P}(D(k)<\gamma|H_1) \geq  \mathbb{P}(D(k) \leq \gamma_1|H_1) \rightarrow \frac{1}{2}$ as
$k \rightarrow \infty$, by the central limit theorem.

Denote by $I_l(\cdot)$, $l=0,1,$ the Fenchel-Legendre transform~\cite{DemboZeitouni} of $\Lambda_l(\cdot)$:
\begin{equation}
\label{eqn-I-0-def}
I_l(z) = \sup_{\lambda \in {\mathbb R}} \lambda z - \Lambda_l(\lambda),\:z \in {\mathbb R}.
\end{equation}
It can be shown~\cite{DemboZeitouni} that $I_l(\cdot)$ is nonnegative, strictly convex, $I_l(\gamma_l)=0$, for $l=0,1$, and $I_1(z) = I_0(z)-z$,~\cite{DemboZeitouni}. We now state
the result on the centralized detector's asymptotic performance.
\begin{lemma}
\label{lemma-centralized-detector}
Let Assumption~\ref{ass_1} hold, and consider the family of centralized detectors \eqref{eq-centralized-dec-rule} with
constant threshold $\gamma=\gamma \in (\gamma_0,\gamma_1).$ Then,
the best (maximal) error exponent:
\[
\lim_{k \rightarrow \infty} -\frac{1}{k} \log P_{\mathrm{e}}(k,\gamma)
\]
is achieved for the zero threshold $\gamma=0$ and equals $N C_{\mathrm{ind}},$ where $C_{\mathrm{ind}} = I_0(0).$
\end{lemma}
The quantity $C_{\mathrm{ind}}$ is referred to as the \emph{Chernoff information} of
a single sensor observation $Y_i(t).$ Lemma \ref{lemma-centralized-detector}
says that the centralized detector' error exponent is $N$ times
larger than an individual sensor's error exponent. We remark that, even if we
allow for time-varying thresholds $\gamma_k=\gamma$, the error
exponent $N C_{\mathrm{ind}}$ cannot be improved, i.e.,
 the centralized detector with zero threshold is asymptotically optimal
 over all detectors. We will see that, when a certain condition
 on the network connectivity holds, the distributed
 detector is asymptotically optimal, i.e., achieves the best error exponent $N C_{\mathrm{ind}}$,
 and the zero threshold is again optimal. However,
 when the network connectivity condition is not met,
 the distributed detector is no longer asymptotically optimal,
 and the optimal threshold may be non zero.
\begin{proof}[Proof of Lemma \ref{lemma-centralized-detector}]
%\begin{proof}
Denote by $\Lambda_{0,N}$ the LMGF for
the log-likelihood ratio $\sum_{i=1}^N L_i(t)$
for the observations of all sensors at time $t$. Then,
$\Lambda_{0,N}(\lambda)=N \Lambda_0(\lambda)$, by the i.i.d. in space assumption on the sensors' observations. The
Lemma now follows by the Chernoff lemma (Corollary 3.4.6,~\cite{DemboZeitouni}):
\begin{equation*}
\lim_{k\rightarrow \infty} -\frac{1}{k}\log P_{\mathrm e}(k,0) = \max_{\lambda \in [0,1]} \left\{- \Lambda_{0,N}(\lambda)
\right \} = N \max_{\lambda \in [0,1]} \left\{- \Lambda_{0}(\lambda)\right\} = N I_0(0).
\end{equation*}
\vspace{-6mm}
\end{proof}

\vspace{-4mm}
\subsection{Distributed detection algorithm}
\label{subsec-Detection-Alg}
We now consider distributed detection when the sensors cooperate through a randomly varying network. Specifically,
  we consider the running consensus distributed detector proposed in~\cite{running-consensus-detection}.
   Each node $i$ maintains its local decision variable $x_i(k)$, which is a local estimate of the global optimal decision variable $D(k)$ in~\eqref{eq-centralized-dec-rule}. Note that $D(k)$ is \textit{not} locally available.
   At each time $k$, each sensor $i$ updates $x_i(k)$ in two ways: 1) by incorporating its new observation $Y_i(k)$
    to make an intermediate decision variable $\frac{k-1}{k}x_i(k-1)+\frac{1}{k}L_i(k)$; and 2)~by exchanging the
    intermediate decision variable
     locally with its neighbors and computing the weighted average
     of its own and the neighbors' intermediate variables.

More precisely, the update of $x_i(k)$ is as follows:
\begin{eqnarray}
\label{eqn-running-cons-sensor-i}
x_i(k) = \sum_{j \in O_i(k)}  W_{ij}(k) \left(  \frac{k-1}{k}x_j(k-1)+\frac{1}{k}L_j(k) \right), \:k=1,2,...\:\:x_i(0) = 0.
\end{eqnarray}
Here $O_i(k)$ is the (random) neighborhood of sensor $i$ at time $k$ (including $i$), and
$ W_{ij}(k)$ are the (random) averaging weights. The sensor~$i$'s local decision test at time $k$ is:
\begin{equation}
\label{eq-distributed-dec-rule}
x_i(k) \stackrel[H_0]{H_1}{\gtrless}\gamma,
\end{equation}
i.e., $H_1$ (respectively, $H_0$) is decided when $x_i(k)\geq \gamma$ (respectively, $x_i(k) < \gamma$).

%\mypar{Remark} The local decision test in~\eqref{eq-distributed-dec-rule}
%uses a zero threshold. It may be
%possible to improve the \emph{non-asymptotic} performance
%of the distributed detector~\eqref{eqn-running-cons-sensor-i} by
%fine tuning the threshold. But because
%our focus is \emph{asymptotic} performance analysis
%and it is well known that, with the centralized detection, the log-likelihood
%ratio test with zero threshold achieves \emph{asymptotic optimality}~\cite{DemboZeitouni},
% we take here the threshold of the sequence of tests in~\eqref{eqn-running-cons-sensor-i} to be zero.
%
%although, for the observation interval of the finite size $k$,
%the threshold in~\eqref{eq-centralized-dec-rule} that minimizes the error probability
%is $\gamma^\star_k=\frac{1}{k} \log(\pi_0/\pi_1)$. Hence,
%the sequences of the decision tests~\eqref{eq-centralized-dec-rule}, for $k=1,2,...$,
%with 1) $\gamma_k=\gamma^\star_k$; and 2) $\gamma_k=0$ are indistinguishable
%in terms of the asymptotic performance (error exponent). The goal
%of this paper is to evaluate the asymptotic performance (error exponent)
% for the distributed detector~\eqref{eqn-running-cons-sensor-i}, and thus,
% in light of the discussion on the centralized detection, we
% select the zero threshold at any $k$. It may be possible to
% improve the non-asymptotic performance of the distributed detector~\eqref{eqn-running-cons-sensor-i}
%  by optimizing the threshold, but this lies outside of our paper's scope.

Write the consensus+innovations algorithm~\eqref{eqn-running-cons-sensor-i} in vector form. Let
$x(k) =
(x_1(k),x_2(k),...,x_N(k))^\top$ and
$L(k)=(L_1(k),...,L_N(k))^\top$.
Also, collect the averaging weights $W_{ij}(k)$ in the $N \times N$
  matrix $W(k)$, where, clearly, $W_{ij}(k)=0$ if the sensors
   $i$ and $j$ do not communicate at time step $k$.
The algorithm~\eqref{eqn-running-cons-sensor-i} becomes:
\begin{eqnarray}
\label{eqn_recursive_algorithm}
x(k) = W(k) \left( \frac{k-1}{k} x(k-1) + \frac{1}{k} L(k) \right),\:k = 1,2,...\:\:x_i(0) = 0.
\end{eqnarray}

%     More precisely, the update of $x_i(k)$ is as follows:
%\begin{eqnarray}
%\label{eqn-running-cons-sensor-i}
%x_i(k) &=& \frac{k-1}{k}  \left( W_{ii}(k)x_i(k-1)+ \sum_{j \in
%O_i(k)}W_{ij}(k) x_j(k-1) \right)+    \frac{1}{k}   \left( W_{ii}(k)L_i(k)+ \sum_{j \in
%O_i(k)}W_{ij}(k) L_j(k) \right), \, \\k&=&1,2,...,\:\:
% x_i(0) =  0. \nonumber
%\end{eqnarray}
%%
%Here $O_i(k)$ is the (random) neighborhood of sensor $i$ at time $k$, and
%$ W_{ij}(k)$ are the (random) averaging weights. The local sensor $i$'s decision test at time $k$ is given by:
%%
%\begin{equation}
%\label{eq-distributed-dec-rule}
%x_i(k) \stackrel[H_0]{H_1}{\gtrless}0,
%\end{equation}
%%
%i.e., $H_1$ (respectively, $H_0$) is decided when $x_i(k)\geq 0$ (resp. $x_i(k) < 0$).
%
%
%Let
%$x(k) =
%(x_1(k),x_2(k),...,x_N(k))^\top$ and
%$L(k)=(L_1(k),...,L_N(k))^\top$.
%Also, collect the averaging weights $W_{ij}(k)$ in $N \times N$
%  matrix $W(k)$, where, clearly, $W_{ij}(k)=0$ if the sensors
%   $i$ and $j$ do not communicate at time step $k$.
%The algorithm in matrix form becomes:
%\begin{eqnarray}
%\label{eqn_recursive_algorithm}
%x(k) = \frac{k-1}{k} W(k) x(k-1) + \frac{1}{k}  W(k)L(k),\,k=1,2,...,\;\; x(0)  =   0.% \nonumber
%\end{eqnarray}

\mypar{Network model}
We  state the assumption on the random averaging matrices
 $W(k)$.
\begin{assumptions}
\label{assumption-W(k)}
The averaging matrices  $W(k)$ satisfy the following:
\begin{enumerate} [(a)]
\item The sequence $\left\{ W(k)  \right\}_{k=1}^{\infty}$ is i.i.d.
\item $W(k)$ is symmetric
 and stochastic (row-sums equal 1 and $W_{ij}(k)\geq 0$) with probability one, $\forall k$.
\item There exists $\eta>0$, such that, for any realization $W(k)$, $W_{ii}(k) \geq \eta$, $\forall i$,
and, $W_{ij}(k) \geq \eta$ whenever $W_{ij}(k)>0$, $i \neq j$.
\item $W(k)$ and $Y(t)$ are mutually independent over all $k$ and $t$.
\end{enumerate}
\end{assumptions}
Condition~(c) is mild and says that: 1) sensor $i$ assigns a
non-negligible weight to itself; and 2) when sensor
$i$ receives a message from sensor $j$, sensor $i$
 assigns a non-negligible weight to sensor~$j$.

%\mypar{} We consider the solution of the consensus+innovations distributed detector~\eqref{eqn_recursive_algorithm}.

Define the matrices $\Phi(k,t)$ by:
\begin{equation}
\label{eqn-def-Phi}
\Phi(k,t):=W(k)W(k-1)...W(t), \:\: k \geq t \geq 1.
\end{equation}
It is easy to verify from~\eqref{eqn_recursive_algorithm} that $x(k)$ equals:
\begin{equation}
\label{alg-unwinded}
x(k) = \frac{1}{k} \sum_{t=1}^{k} \Phi(k,t) L(t),\,\,k=1,2,...
\end{equation}

\mypar{Choice of threshold $\gamma$} We restrict the choice of threshold $\gamma$ to
$\gamma \in (\gamma_0,\gamma_1)$, $\gamma_0<0$, $\gamma_1>0$, where we recall
$\gamma_l = \mathbb{E}[L_1(1)|H_l]$, $l=0,1.$ Namely,
$W(t)$ is a stochastic matrix, hence $W(t)1=1$, for all $t$, and
thus $\Phi(k,t) 1 =1$. Also, $\mathbb{E}[L(t)|H_l] =\gamma_l 1$, for all $t$, $l=0,1$. Now, by iterating expectation:
\begin{eqnarray*}
\mathbb{E}[x(k)|H_l] = \mathbb{E} [ \mathbb{E}[x(k)|H_l, W(1),...,W(k)]]
 = \mathbb{E} \left[\frac{1}{k} \sum_{t=1}^{k} \Phi(k,t) \mathbb{E}[L(t)|H_l] \right]= \gamma_l 1,\:l=0,1,
\end{eqnarray*}
and $\mathbb{E}[x_i(k)|H_l]=\gamma_l $, for all $i,k$. Moreover,
it can be shown (proof is omitted due to lack of space)
 that $x_i(k)$ converges in probability
 to $\gamma_l$ under $H_l$. Now,
 a similar argument as with the centralized detector in \ref{subsec-Centralized-detection}
  shows that for $\gamma \notin (\gamma_0,\gamma_1)$,
  the error probability does not converge to zero.
  We will show that, for any $\gamma \in (\gamma_0,\gamma_1)$,
  the error probability converges
  to 0 exponentially fast, and we
  find the optimal $\gamma=\gamma^\star$ that maximizes
  a certain lower bound on the exponent of the error probability.

\mypar{Network connectivity} From~\eqref{alg-unwinded},
we can see that the matrices
$\Phi(k,t)$ should be as close to $J$ as possible for enhanced detection performance.
Namely, the ideal (unrealistic)
 case when $\Phi(k,t) \equiv J$ for all $k,t$, corresponds
 to the scenario where each sensor $i$ is equivalent to the
 optimal centralized detector. It is well known that, under certain conditions,
the matrices $\Phi(k,t)$ converge in probability to $J$:
\begin{equation*}
\label{eqn-pomocna}
\mathbb {P}\left(\|\Phi(k,t)-J\|> \epsilon\right) \rightarrow 0 \: \mathrm{as}\: (k-t) \rightarrow \infty,\: \epsilon>0,
\end{equation*}
such that $\mathbb {P}\left(\|\Phi(k,t)-J\|> \epsilon\right)$ vanishes
exponentially fast in $(k-t)$, i.e.,
$\mathbb {P}\left(\|\Phi(k,t)-J\|> \epsilon\right) \approx r^{(k-t)}$, $r \in [0,1]$.
The quantity $r$
 determines the speed of convergence of the matrices $\Phi(k,t)$.
The closer to zero $r$ is, the faster consensus is. We refer to $|\log r|$ as the network connectivity. We will see
that the distributed detection performance significantly depends on $r$.
Formally, $|\log r| = -\log r$ is given by:\footnote{It can be shown
that the limit in \eqref{eqn-r-def} exists and that it does not depend on $\epsilon$.}
\begin{equation}
\label{eqn-r-def}
|\log r| : = \lim_{(k-t) \rightarrow \infty} -\frac{1}{k-t} \log \mathbb {P}\left(\|\Phi(k,t)-J\|> \epsilon\right).
\end{equation}
For the exact calculation of $r$, we refer to~\cite{rate-of-consensus}. Reference~\cite{rate-of-consensus} shows that, for the commonly used models of $W(k)$,
gossip and link failure (links in the underlying network fail independently, with possibly mutually different probabilities),
$r$ is easily computable, by solving a certain min-cut problem.
In general, $r$ is not easily computable, but all our results
 (Theorem \ref{Theorem-rate-bounds-for-alpha-and-beta}, Corollary \ref{Corollary-main-result}, Corollary \ref{Corollary-main-result-non-id}) hold when $r$ is replaced by an upper bound.
 An upper bound on $r$ is given by $\mu_2 \left( \mathbb{E} \left[  W^2(k)\right] \right)$,~\cite{rate-of-consensus}.

The following Lemma easily follows from~\eqref{eqn-r-def}.
\begin{lemma}
\label{lemma-conv-of-Phi}
Let Assumption \ref{assumption-W(k)} hold. Then, for any $\delta>0$,
 there exists a constant $C(\delta) \in (0,\infty)$ (independent of $\epsilon \in (0,1)$) such that:
 \begin{equation*}
 %\label{eqn-Prob-for-proof}
 \mathbb {P}\left(\|\Phi(k,t)-J\|> \epsilon\right) \leq C(\delta) e^{-(k-t) (|\log r| - \delta)},\:\mathrm{for \,\,all\,\,}k\geq t.
 \end{equation*}
 \end{lemma}
 %
 %

 %
%
%characterize the speed of consensus across the network as in~\cite{weight-opt}:
%\begin{equation*}
%r:=\mu_2\left(\E\left[ W^2(k)\right]\right).
%\end{equation*}
%The quantity $r \in [0,1]$, and, when $r<1$, the standard consensus algorithm converges in mean square sense, see~\cite{weight-opt}. As shown in~\cite{weight-opt}, the smaller $r$ is, the faster consensus is.
%
% The convergence of standard consensus is equivalent to
% the convergence of $\Phi(k,t)$ to the ideal consensus matrix $J$, i.e.,
% \[
% \Phi(k,t) \rightarrow J=(1/N)11^\top\;\;\mathrm{as}\;\; (k-t) \rightarrow \infty.
% \]
%  The next lemma provides a bound on the
%  \emph{rate of convergence in probability} of the sequence of matrices $\Phi(k,t)$ in terms of $r$; it will be of interest in the next section. The Lemma is proved in~\cite{GaussianDD}.
%  %
%  %
%\begin{lemma}\label{lemma-conv-of-Phi}
%Let Assumption~3 hold. For $\epsilon>0$, $\Phi(k,t)$ as in~\eqref{eqn-def-Phi}, and fixed $k\geq t\geq 1$:
%\begin{equation*}
%%\label{eq-conv-of-Phi}
%\mathbb {P}\left(\|\Phi(k,t)-J\|> \epsilon\right)\leq \frac{N^4}{\epsilon^2} r^{k-t}.
%\end{equation*}
%\end{lemma}
%\begin{proof} See~\cite{GaussianDD}.
%\end{proof}
%
%
%
\section{Main results: Asymptotic analysis and error exponents for distributed detection}
\label{Main-result}
Subsection \ref{subsect-statement-of-results} states our main results
on the asymptotic performance of consensus+innovations distributed detection;
 subsection \ref{subsect-proofs} proves these results.
\subsection{Statement of main results}
\label{subsect-statement-of-results}
In this section, we analyze the performance of distributed detection in terms of the detection error exponent, when the number of observations (per sensor), or the size $k$ of the observation interval tends to $+\infty$. As we will see next, we show that there exists a threshold on the network connectivity $|\log r|$ such that if $|\log r|$ is above this threshold, each node in the network achieves asymptotic optimality (i.e., the error exponent at each node is the total Chernoff information equal to $N C_{\mathrm {ind} }$). When $|\log r|$ is below the threshold, we give a lower bound for the error exponent. Both the threshold and the lower bound are given \emph{solely} in terms of the log-moment generating function $\Lambda_0$ and the number of sensors $N$.
These findings are summarized in Theorem~\ref{Theorem-rate-bounds-for-alpha-and-beta} and Corollary~\ref{Corollary-main-result} below.

Let $\alpha_i(k, \gamma)$, $\beta_i(k, \gamma)$, and $P_{\mathrm e,i} (k, \gamma)$ denote the probability of false alarm, the probability of miss,
and the error probability, respectively, of sensor $i$ for the detector~\eqref{eqn-running-cons-sensor-i} and~\eqref{eq-distributed-dec-rule}, for the threshold equal to $\gamma$:
\begin{eqnarray}
\label{eqn-P-fa}
\alpha_i(k, \gamma)=\mathbb{P} \left(x_i(k)\geq \gamma|H_0\right),\:
\beta_i(k, \gamma)=\mathbb{P} \left(x_i(k)< \gamma|H_1\right),\:
P_{\mathrm e,i} (k, \gamma) = \pi_0 \alpha_i(k; \gamma)+\pi_1 \beta_i(k; \gamma),
\end{eqnarray}
where, we recall, $\pi_1$ and $\pi_0$ are the prior probabilities.

\begin{theorem}
\label{Theorem-rate-bounds-for-alpha-and-beta}
Let Assumptions 1-3 hold and consider the family of distributed detectors in~\eqref{eqn-running-cons-sensor-i} and~\eqref{eq-distributed-dec-rule} with $\gamma \in (\gamma_0,\gamma_1)$. Let $\lambda_l^{\mathrm{s}}$
be the zero of the function:
%\footnote{Further ahead, in the proof
%of Corollary~\ref{Corollary-main-result}, we
% show that, if $|\log r|<\mathrm{thr}\left(\Lambda_0, N\right)$,
% the function $\Delta_l(\cdot)$ has a zero; the set of zeros of $ \Delta_l(\cdot)$
% is either a closed interval, or a singleton.}
 \begin{equation}
 \label{def-Delta}
 \Delta_l(\lambda): = \Lambda_l(N \lambda)-|\log r|-N\Lambda_l(\lambda),\:\:l=0,1,
 \end{equation}
and define $\gamma_l^{-}, \gamma_l^{+}$, $l=0,1$ by
\begin{align}
\label{eqn-gamma-plus-minus}
\gamma_0^{-}&=\Lambda_0^\prime(\lambda_0^{\mathrm{s}}),\;\gamma_0^{+}=\Lambda_0^\prime(N\lambda_0^{\mathrm{s}})\geq \gamma_0^{-}\\
\gamma_1^{-}&=\Lambda_1^\prime(N\lambda_1^{\mathrm{s}}),\;\gamma_1^{+}=\Lambda_1^\prime(\lambda_1^{\mathrm{s}})\geq \gamma_1^{-}.
\end{align}
Then, for every $\gamma\in (\gamma_0,\gamma_1)$,
at each sensor~$i$, $i=1,\ldots,N$, we have:
\begin{eqnarray}
\label{theorem-rate-bound-for-alpha}
\liminf_{k\rightarrow \infty} -\frac{1}{k}\log \alpha_i(k,\gamma) \geq  B_0(\gamma),
\:\:\:\:\:
%\\
%\label{theorem-rate-bound-for-beta}
\liminf_{k\rightarrow \infty}-\frac{1}{k}\log \beta_i(k,\gamma)\geq B_1(\gamma),
\end{eqnarray}
where
\begin{align*}
%\label{eq-B-0}
B_0(\gamma)&=\hspace{-1mm}\max_{\lambda\in [0,1]} N\gamma\lambda-\max\{N\Lambda_0(\lambda), \Lambda_0(N\lambda)-|\log r| \}=\left\{ \begin{array}{lll} NI_0(\gamma),  &  \gamma\in (\gamma_0,\gamma_0^{-}] \\
NI_0(\gamma_0^{-})+N\lambda_0^{\mathrm{s}} (\gamma-\gamma_0^{-}), &   \gamma\in (\gamma_0^{-},\gamma_0^{+})\\
I_0(\gamma)+|\log r|, &   \gamma\in [\gamma_0^{+},\gamma_1) \end{array} \right. \\
%\label{eq-B-0}
B_1(\gamma)&= \hspace{-2mm}\max_{\lambda\in [-1,0]} N\gamma\lambda-\max\{N\Lambda_1(\lambda), \Lambda_1(N\lambda)-|\log r| \}=  \left\{ \begin{array}{lll} I_1(\gamma)+|\log r|, &    \gamma\in (\gamma_0,\gamma_1^{-}] \\
NI_1(\gamma_1^{+}) + N\lambda_1^{\mathrm{s}} (\gamma-\gamma_1^{+}), &    \gamma\in (\gamma_1^{-},\gamma_1^{+})\\
NI_1(\gamma), &    \gamma\in [\gamma_1^{+},\gamma_1) .\end{array} \right.
\end{align*}
\end{theorem}
%
%
%\begin{theorem}
%\label{Theorem-rate-bounds-for-alpha-and-beta}
%Let Assumptions 1-3 hold and consider the distributed detector in~\eqref{eqn-running-cons-sensor-i}. For every $\lambda\geq 0$,
%at each sensor~$i$, $i=1,\ldots,N$, we have:
%%
%\begin{eqnarray}
%\label{theorem-bound-on-alpha}
%\liminf_{k\rightarrow \infty} -\frac{1}{k}\log \alpha_i(k) &\geq&-\max \left\{N \Lambda_0(\lambda), \Lambda_0(N \lambda)-|\log r|\right\}\\
%\label{theorem-bound-on-beta}
%\liminf_{k\rightarrow \infty}-\frac{1}{k}\log \beta_i(k)&\geq& -\max \left\{N \Lambda_0(1-\lambda), \Lambda_0(1-N \lambda)-|\log r|\right\}.
%\end{eqnarray}
%%
%\end{theorem}
%
%
%

\begin{corollary}
\label{Corollary-main-result}
Let Assumptions 1-3 hold and consider the family of distributed detectors in~\eqref{eqn-running-cons-sensor-i} and~\eqref{eq-distributed-dec-rule} parameterized by detector thresholds $\gamma \in (\gamma_0,\gamma_1)$. Then:
\begin{enumerate}[(a)]
\item
\begin{eqnarray}
\label{eqn-rate}
\liminf_{k \rightarrow \infty} -\frac{1}{k} \log P_{\mathrm e,i}(k,\gamma)
 \geq \min\{B_0(\gamma), B_1(\gamma)\}>0,
\end{eqnarray}
and the lower bound in \eqref{eqn-rate} is maximized for the point $\gamma^\star \in (\gamma_0,\gamma_1)$\footnote{As we show in the proof, such a point exists and is unique.} at which $B_0(\gamma^\star)=B_1(\gamma^\star).$
\item Consider $\lambda^\bullet = \mathrm{arg\,min}_{\lambda \in {\mathbb R}} \Lambda_0(\lambda)$, and let:
\begin{eqnarray}
\label{eqn-threshold}
\mathrm{thr}\left(\Lambda_0, N\right) =  \max \{\Lambda_0(N \lambda^\bullet)-N\Lambda_0(\lambda^\bullet),
\Lambda_0(1-N(1- \lambda^\bullet))-N\Lambda_0(\lambda^\bullet)\},
\end{eqnarray}
Then, when $|\log r|\geq \mathrm{thr}(\Lambda_0,N)$, each sensor $i$ with the detector threshold set to $\gamma=0$, is asymptotically optimal:
 \[
 \lim_{k \rightarrow \infty}-\frac{1}{k} \log P_{\mathrm e,i}(k,0) = N C_{\mathrm{ind}}.
 \]
%
%is asymptotically optimal at each sensor~$i$.
\item When
$\Lambda_0(\lambda)=\Lambda_0(1-\lambda)$,
for $\lambda \in [0,1]$ $\gamma^\star=0$, irrespective of the value of $r$ (even when $|\log r|<\mathrm{thr}(\Lambda_0,N)$.)
% This symmetry holds, e.g., for the Gaussian and Laplace distribution
% considered in Section \ref{sec-Examples}.
\end{enumerate}
\end{corollary}
Figure 1 (left) illustrates the error exponent lower bounds $B_0(\gamma)$
       and $B_1(\gamma)$ in Theorem~\ref{Theorem-rate-bounds-for-alpha-and-beta}, 
       while Figure~1~(right) illustrates the quantities in \eqref{eqn-gamma-plus-minus}. ( See the definition of the function $\Phi_0(\lambda)$ in~\eqref{eqn-Phi-0-def} in the proof of Theorem \ref{Theorem-rate-bounds-for-alpha-and-beta}.) We consider $N=3$ sensors and a discrete distribution of $Y_i(t)$ over a 5-point alphabet, with the distribution $[.2, .2, .2, .2, .2]$ under $H_1$, and $[0.01, 0.01, 0.01, 0.01, 0.96]$ under $H_0$. We set
         here $r=0.4.$
         \vspace{-5mm}
         \begin{figure}[thpb]
      \centering
      \includegraphics[height=2.1 in,width=3.1 in]{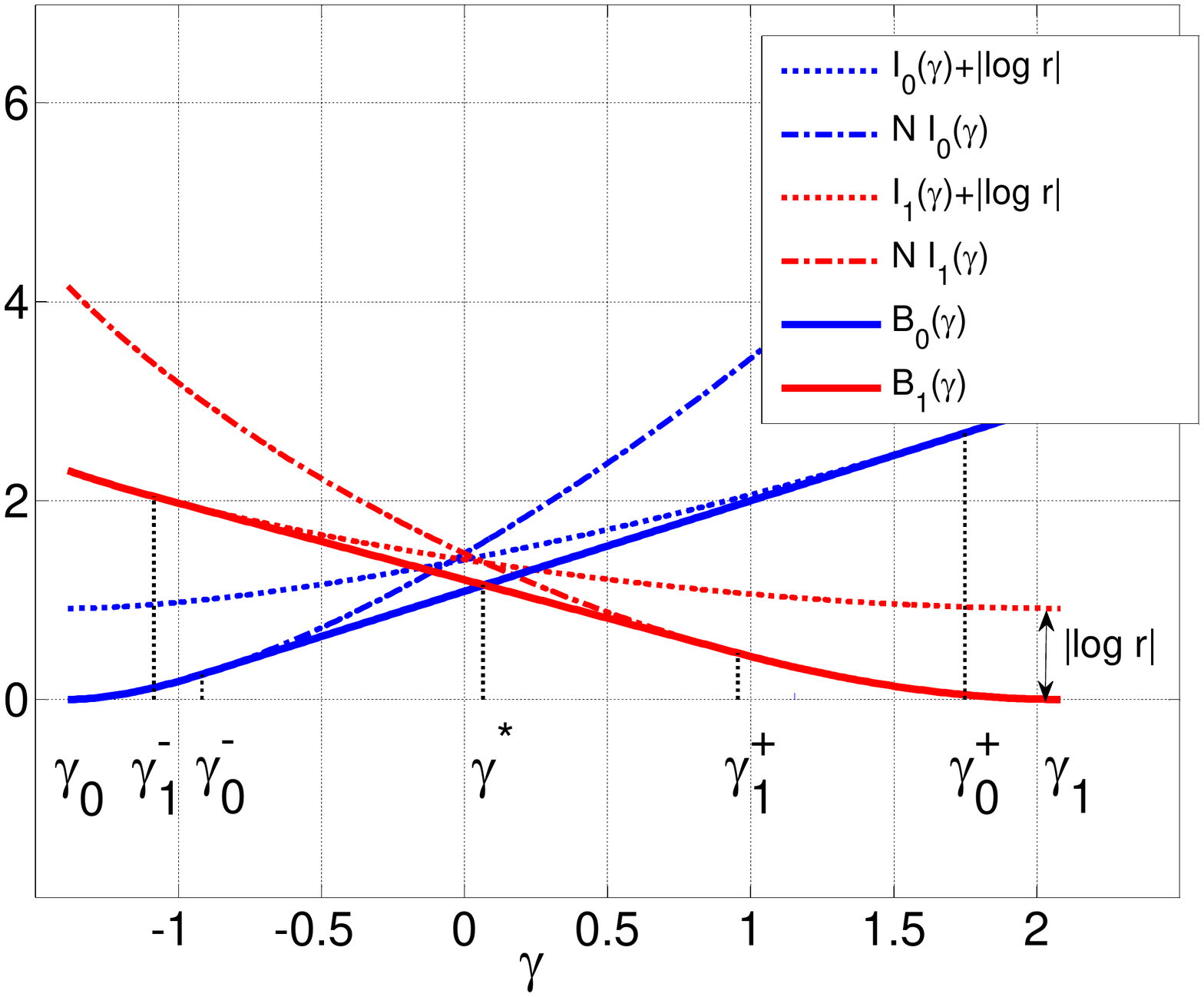}
      \includegraphics[height=2.1 in,width=3.1 in]{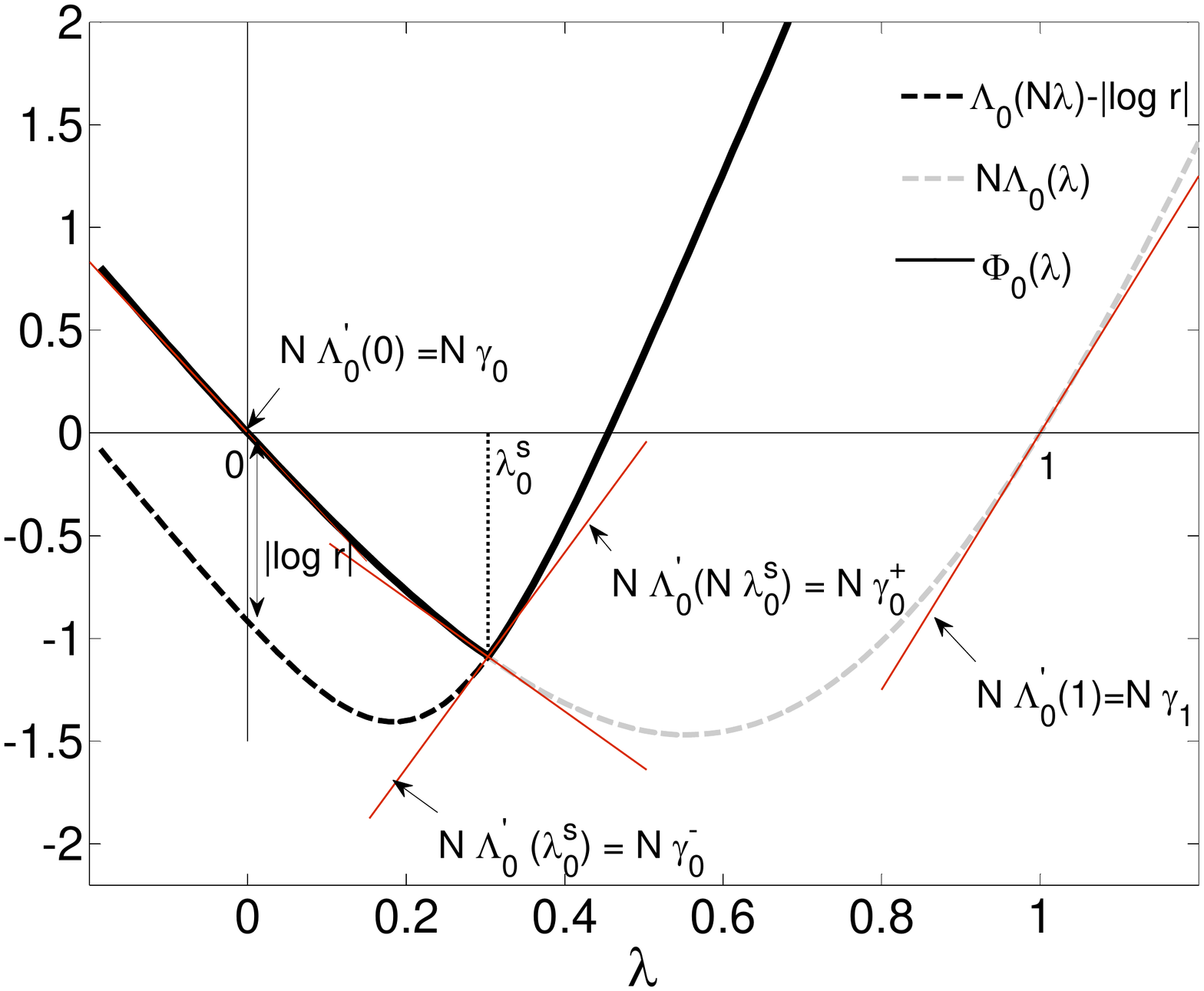}
      \caption{\textbf{Left:} Illustration of the error exponent lower bounds $B_0(\gamma)$
       and $B_1(\gamma)$ in Theorem~\ref{Theorem-rate-bounds-for-alpha-and-beta}; \textbf{Right:}
        Illustration of the function $\Phi_0(\lambda)$ in~\eqref{eqn-Phi-0-def},
        and the quantities in \eqref{eqn-gamma-plus-minus}. We consider $N=3$ sensors and a discrete distribution of $Y_i(t)$ over a 5-point alphabet, with the distribution $[.2, .2, .2, .2, .2]$ under $H_1$, and $[0.01, 0.01, 0.01, 0.01, 0.96]$ under $H_0$. We set
         here $r=0.4.$}
      \label{Fig-corollary}
\end{figure}

 Corollary~\ref{Corollary-main-result} states that, when the network connectivity
 $|\log r|$ is above a threshold, the distributed detector in \eqref{eqn-running-cons-sensor-i} and~\eqref{eq-distributed-dec-rule}
  is asymptotically equivalent to the optimal centralized detector. The corresponding optimal detector threshold is $\gamma=0$.
When $|\log r|$ is below the threshold, Corollary~\ref{Corollary-main-result} determines
   what value of the error exponent the distributed detector can achieve, for any given $\gamma \in (\gamma_0,\gamma_1)$. Moreover, Corollary~\ref{Corollary-main-result} finds the optimal detector threshold $\gamma^\star$ for a given $r$; $\gamma^\star$ can be found as the unique zero of the strictly decreasing function $\Delta_B(\gamma):=B_1(\gamma)-B_0(\gamma)$ on $\gamma\in (\gamma_0,\gamma_1)$, see the proof of Corollary~\ref{Corollary-main-result}, e.g., by bisection on $(\gamma_0,\gamma_1)$.

\mypar{Remark} When
$\Lambda_0(\lambda)=\Lambda_0(1-\lambda)$,
for $\lambda \in [0,1]$,
it can be shown that $\gamma_0 = -\gamma_1<0$,
and $B_0(\gamma)=B_1(-\gamma)$,
for all $\gamma \in (\gamma_0,\gamma_1)$.
This implies that the point
$\gamma^\star$ at which $B_0$ and $B_1$ are equal
is necessarily zero, and hence the optimal
detector threshold $\gamma^\star=0$, irrespective of
the network connectivity $|\log r|$ (even when $|\log r|<\mathrm{thr}(\Lambda_0,N)$.)
 This symmetry holds, e.g., for the Gaussian and Laplace distribution
 considered in Section \ref{sec-Examples}.

   Corollary~\ref{Corollary-main-result} establishes that
    there exists a ``sufficient'' connectivity, say $|\log r^\star|$,
    so that further improvement on the connectivity (and further spending of resources, e.g., transmission power)
     does not lead to a pay off in terms of detection performance. Hence,
    Corollary~\ref{Corollary-main-result} is valuable
    in the practical design of a sensor network,
    as it says how much connectivity (resources) is sufficient to achieve asymptotically optimal detection.

    Equation~\eqref{eqn-rate} says that the distribution of the sensor observations
    (through LMGF) plays a role in determining the performance of distributed detection.
    We illustrate and explain by examples this effect in Section~\ref{sec-Examples}.

\subsection{Proofs of the main results}
\label{subsect-proofs}
%
%Corollary~\ref{Corollary-main-result} follows from Theorem~\ref{Theorem-rate-bounds-for-alpha-and-beta} by optimizing the bounds~\eqref{theorem-bound-on-alpha} and~\eqref{theorem-bound-on-beta}.
We first prove Theorem~\eqref{Theorem-rate-bounds-for-alpha-and-beta}.
\begin{proof}[Proof of Theorem \ref{Theorem-rate-bounds-for-alpha-and-beta}]
Consider the probability of false alarm $\alpha_i(k,\gamma)$ in~\eqref{eqn-P-fa}.
 We upper bound $\alpha_i(k,\gamma)$ using the exponential Markov inequality~\cite{Karr}
 parameterized by $\zeta \geq 0$:
 \begin{eqnarray}
 \label{eq-Exp-Markov-inequality}
 \alpha_i(k,\gamma) =  \mathbb{P} \left( x_i(k) \geq \gamma\,|\,H_0 \right)
 =  \mathbb{P} \left( e^{\zeta x_i(k)} \geq e^{\zeta \gamma}\,|\,H_0 \right)
 \leq \E\left[e^{\zeta x_i(k)}|H_0\right] e^{-\zeta \gamma}.
 \end{eqnarray}
Next, by setting $\zeta = N\,k\,\lambda$, with $\lambda \geq 0$, we obtain:
\begin{eqnarray}
\label{eq-Exp-Markov}
\alpha_i(k,\gamma)& \leq & \E\left[e^{Nk\lambda x_i(k)}|H_0\right] e^{-N k \lambda \gamma} \\
\label{eq-exp-markov-2}
&=&\E\left[e^{N\lambda  \sum_{t=1}^k\sum_{j=1}^N \Phi_{i,j}(k,t) L_j(t)}|H_0\right]e^{-N k \lambda \gamma} .
\end{eqnarray}
The terms in the sum in the exponent in~\eqref{eq-exp-markov-2} are conditionally independent, given the realizations of the averaging matrices $W(t)$, $t=1,\ldots,k$, Thus, by iterating the expectations, and using the
definition of $\Lambda_0$ in~\eqref{eq-log-momentgen-fcn}, we compute the expectation in~\eqref{eq-exp-markov-2} by
conditioning first on $W(t)$, $t=1,\ldots,k$:
\begin{eqnarray}
\label{eq-Conditioning-on-Wl}
\E\left[e^{N\lambda  \sum_{t=1}^k\sum_{j=1}^N \Phi_{i,j}(k,t) L_j(t)}|H_0\right]&=&\E\left[\E\left[e^{N\lambda  \sum_{t=1}^k\sum_{j=1}^N \Phi_{i,j}(k,t) L_j(t)}|H_0,W(1),\ldots,W(k)\right]\right]\nonumber \\
&=& \E\left[e^{\sum_{t=1}^k\sum_{j=1}^N \Lambda_0\left(N\lambda \Phi_{i,j}(k,t)\right)}\right].
\end{eqnarray}

\mypar{Partition of the sample space} We handle the random matrix realizations $W(t)$, $t=1,\ldots,k$,
through a suitable partition of the underlying probability space. Adapting an argument from~\cite{GaussianDD},
partition the probability
space \emph{based on the time of the last successful averaging}. In more detail,
for a fixed $k$, introduce the partition $\mathcal {P}_k$ of the sample space that
consists of the disjoint events $\mathcal {A}_{s,k}$, $s=0,1,...,k$, given by:
\begin{equation*}
\mathcal {A}_{s,k} = \left\{\|\Phi(k,s)-J\|\leq \epsilon \; \mathrm{and}\; \|\Phi(k,s+1)-J\|>\epsilon  \right\},
\end{equation*}
for $s=1,...,k-1$, $\mathcal {A}_{0,k}=\{\|\Phi(k,1)-J\|>\epsilon \}$,
and ${A}_{k,k}=\left\{\|\Phi(k,k)-J\|\leq \epsilon\right\}$. For simplicity of notation, we drop the index $k$ in the sequel and denote event $\mathcal A_{s,k}$ by $\mathcal A_s$, $s=0,\ldots,k$.
%
%\begin{eqnarray*}
%\mathcal {A}_0&=&\{\|\Phi(k,1)-J\|>\epsilon \},\\
%\mathcal {A}_1&=&\left\{\|\Phi(k,1)-J\|\leq \epsilon \; \mathrm{and}\; \|\Phi(k,2)-J\|>\epsilon  \right\},\\
%\ldots \\
%\mathcal {A}_s&=&\left\{\|\Phi(k,s)-J\|\leq \epsilon \; \mathrm{and}\; \|\Phi(k,s+1)-J\|>\epsilon  \right\},\\
%\ldots \\
%\mathcal {A}_k&=&\left\{\|\Phi(k,k)-J\|\leq \epsilon\right\},
%\end{eqnarray*}
%
for $\epsilon>0$. Intuitively,
the smaller $t$ is, the closer
the product $\Phi(k,t)$ to $J$ is; if the event $\mathcal{A}_s$ occurred,
then the
largest $t$ for which
the product $\Phi(k,t)$ is still $\epsilon$-close to
$J$ equals $s$. We now show that $\mathcal{P}_k$ is indeed a partition.
We need the following simple Lemma. The Lemma
shows that convergence of $\Phi(k,s)-J$ is monotonic, for any
realization of the matrices $W(1),W(2),...,W(k).$
 \begin{lemma}
 \label{lemma-new-aux}
 Let Assumption~\ref{assumption-W(k)} hold. Then, for any realization of
 the matrices $W(1),...,W(k)$:
\[
\|\Phi(k,s) - J\| \leq \| \Phi(k,t) - J\|,\:\mathrm{for\,\,} 1 \leq s \leq t \leq k.
\]
%\end{enumerate}
\end{lemma}
\begin{proof} Since every realization of $W(t)$ is stochastic and symmetric for every $t$, we have that $W(t)1=1$ and $1^\top W(t)=1^\top$, and, so: $\Phi(k,s) - J=W(k)\cdots W(s)-J=(W(k)-J)\cdots(W(s)-J) $. Now, using the sub-multiplicative property of the spectral norm, we get
\begin{align*}
\|\Phi(k,s) - J\|&=\|(W(k)-J)\cdots(W(t)-J)(W(t-1)-J)\cdots(W(s)-J)\|\\
&\leq \|(W(k)-J)\cdots(W(t)-J)\|\|(W(t-1)-J)\|\cdots\|(W(s)-J)\|.
\end{align*}
To prove Lemma~\ref{lemma-new-aux}, it remains to show that $\|W(t)-J\|\leq 1$, for any realization of $W(t)$. To this end, fix a realization $W$ of $W(t)$. Consider the eigenvalue decomposition $W=QMQ^\top$, where $M=\mathrm{diag}(\mu_1,\ldots,\mu_N)$ is the matrix of eigenvalues of $W$, and the columns of $Q$ are the orthonormal eigenvectors.
As $\frac{1}{\sqrt{N}}1$ is the eigenvector associated with eigenvalue $\mu_1=1$,
we have that $W - J = Q M^\prime Q^\top,$
 where $M=\mathrm{diag}(0,\mu_2,\ldots,\mu_N)$. Because
 $W$ is stochastic, we know that $1=\mu_1 \geq \mu_2 \geq ... \geq \mu_N\geq -1$,
 and so $\|W-J\|  = \max\{|\mu_2|,|\mu_N|\} \leq 1.$
\end{proof}
To show that $\mathcal{P}_k$ is a partition, note first
that (at least) one of the events $\mathcal{A}_0,...,\mathcal{A}_k$ necessarily
occurs. It remains to show that the events $\mathcal{A}_s$ are disjoint. We
carry out this by fixing arbitrary $s=1,...,k$, and showing that, if the event $\mathcal{A}_s$
  occurs, then $\mathcal{A}_t$, $t \neq s$, does not occur.
  Suppose that $\mathcal{A}_s$ occurs, i.e.,
  the realizations $W(1),...,W(k)$ are such that
  $\|\Phi(k,s)-J\|\leq \epsilon$ and $\|\Phi(k,s+1)-J\|>\epsilon$.
  Fix any $t>s.$ Then, event $\mathcal{A}_t$
   does not occur, because, by Lemma \ref{lemma-new-aux},
   $\|\Phi(k,t)-J\| \geq \|\Phi(k,s+1)-J\|>\epsilon.$ Now, fix any
   $t<s.$ Then, event $\mathcal{A}_t$ does not occur,
   because, by Lemma \ref{lemma-new-aux}, $\|\Phi(k,t+1) - J\| \leq \|\Phi(k,s) - J\| \leq \epsilon.$ Thus,
   for any $s=1,...,k$, if the event $\mathcal{A}_s$ occurs,
   then $\mathcal{A}_t$, for $t \neq s$, does not occur, and hence the events
   $\mathcal{A}_s$ are disjoint.

Using the
total probability law over $\mathcal {P}_k$, the expectation~\eqref{eq-Conditioning-on-Wl} is computed by:
\begin{eqnarray}
\label{eq-Condition-on-As}
\E\left[e^{\sum_{t=1}^k\sum_{j=1}^N \Lambda_0\left(N\lambda \Phi_{i,j}(k,t)\right)}\right]=
\sum_{s=0}^{k}\E\left[e^{\sum_{t=1}^k\sum_{j=1}^N \Lambda_0\left(N\lambda \Phi_{i,j}(k,t)\right)} \, \mathcal{I}_{\mathcal{A}_s}\right],
\end{eqnarray}
where, we recall, $\mathcal{I}_{\mathcal{A}_s}$ is the indicator function of the event $\mathcal{A}_s$.
The following lemma explains how to use the partition $\mathcal {P}_k$ to upper bound the expectation in~\eqref{eq-Condition-on-As}.
\begin{lemma}\label{lemma-Bounds-on-Lambda}
Let Assumptions 1-3 hold. Then:
\begin{enumerate}[(a)]
\item \label{lemma-Bounds-on-Lambda-part-b}
For any realization of the random matrices $W(t)$, $t=1,2,...,k$:
  \[
  \sum_{j=1}^N\Lambda_0\left(N\lambda \Phi_{i,j}(k,t)\right)\leq  \Lambda_0\left(N\lambda\right),
  \: \: \forall t=1, \ldots, k.\]
\item \label{lemma-Bounds-on-Lambda-part-a}
Further, consider a fixed $s$ in $\{0,1,...,k\}$. If the
event $\mathcal{A}_s$ occurred, then, for $i=1,\ldots,N$:
$
\Lambda_0\left(N\lambda \Phi_{i,j}(k,t)\right)\leq \max\left( \Lambda_0\left(\lambda-\epsilon N \sqrt {N}\lambda\right), \Lambda_0\left(\lambda+\epsilon N \sqrt {N}\lambda\right) \right),\: \forall t=1,\ldots,s, \,\forall j=1,\ldots,N.$
\end{enumerate}
\end{lemma}
\begin{proof}
To prove part~\eqref{lemma-Bounds-on-Lambda-part-b} of the Lemma, by convexity of $\Lambda_0$, the maximum of $\sum_{j=1}^N \Lambda_0(N\lambda a_j)$ over the simplex $\left\{a \in \mathbb {R}^N:\, \sum_{j=1}^N a_j=1,\, a_j\geq 0, \, j=1,\ldots,N  \right\}$ is achieved at a corner point of the simplex.
  The maximum equals:
  $
  \Lambda_0(N\lambda)+(N-1) \Lambda_0(0) = \Lambda_0(N\lambda),
  $
where we use the property
 from Lemma~\ref{lemma-log-momentgen-fcns}, part~(b), that $\Lambda_0(0)=0$. Finally, since for any realization of the matrices $W(1),\ldots, W(k)$, the set of entries $\left\{\Phi_{i,j}(k,t): j=1,\ldots,N\right\}$ is a point in the simplex, the claim of part~\eqref{lemma-Bounds-on-Lambda-part-b} of the Lemma follows.

To prove part~\eqref{lemma-Bounds-on-Lambda-part-a} of the Lemma, suppose that event $\mathcal{A}_s$ occurred. Then, by the definition of $\mathcal{A}_s$,
\[
\|\Phi(k,s)-J\|=\|W(k)\cdot \ldots \cdot W(s)-J\|\leq \epsilon.
\]
Using the fact that each realization $W(t)$, $t=1,2,\ldots$, is doubly stochastic, and using the sub-multiplicative property of the spectral norm, we have that
\[
\|\Phi(k,t)-J\|=\|W(k)\cdot \ldots \cdot W(t)-J\|\leq \epsilon,
 \]
 for every $t\leq s$. Then, by the equivalence of the 1-norm and the spectral norm, it follows that:
\begin{equation*}
%\label{eqn-l1-bounding}
\left|\Phi_{i,j}(k,t)-\frac{1}{N}\right|\leq \sqrt {N}\epsilon, \;\mathrm{for}\; t=1,\ldots,s,\,\, \mathrm{for\,\,all}\,\,i,j=1,\ldots,N.
\end{equation*}
Finally, since $\Lambda_0$ is convex (Lemma~\ref{lemma-log-momentgen-fcns}, part~(a)), its maximum in~$\left[ \lambda - \epsilon N \sqrt {N}\lambda, \lambda + \epsilon N \sqrt {N}\lambda\right]$ is attained at a boundary point and the claim follows.
\end{proof}
We now fix $\delta \in (0, |\log r|)$. Using the results from Lemma~\ref{lemma-conv-of-Phi} and Lemma~\ref{lemma-Bounds-on-Lambda}, we next bound the expectation in~\eqref{eq-Condition-on-As} as follows:
\begin{eqnarray}
\sum_{s=0}^{k}\E\left[e^{\sum_{t=1}^k\sum_{j=1}^N \Lambda_0\left(N\lambda \Phi_{i,j}(k,t)\right)} \,
\mathcal{I}_{\mathcal{A}_s} \right]
%\mathbb{P}\left(\mathcal {A}_s\right)
&\leq&
 \sum_{s=0}^{k} \left( e^{sN \max\left( \Lambda_0\left(\lambda-\epsilon N \sqrt {N}\lambda\right), \Lambda_0\left(\lambda+\epsilon N \sqrt {N}\lambda\right)
 \right)+(k-s)\Lambda_0(N\lambda)} \right) \nonumber\\
 \label{eq-Main-bound}
 &\times& \left( C(\delta) e^{-(k-(s+1)) (|\log r| - \delta)}\right).
\end{eqnarray}
To simplify the notation, we introduce the function:
\begin{equation}
\label{eq-f-epsilon}
g_0: \mathbb {R}^2\longrightarrow \mathbb{R}, \;g_0( \epsilon, \lambda):=\max\left( \Lambda_0\left(\lambda-\epsilon N \sqrt {N}\lambda\right), \Lambda_0\left(\lambda+\epsilon N \sqrt {N}\lambda\right)
 \right).
\end{equation}
We need the following property of $g_0(\cdot,\cdot)$.
\begin{lemma}\label{lemma-inf-f-epsilon}
Consider $g_0(\cdot,\cdot)$ in~\eqref{eq-f-epsilon}. Then, for every $\lambda \in \mathbb R$, the following holds:
\begin{equation*}
\label{eq-inf-f-epsilon}
\inf_{\epsilon>0} g_0(\epsilon,\lambda)=\Lambda_0(\lambda).
\end{equation*}
\end{lemma}
\begin{proof}
Since $\Lambda_0(\cdot)$ is convex, for $\epsilon'<\epsilon$ and for fixed $\lambda$, we have that \[g_0(\epsilon,\lambda)=\max_{\delta \in [-\epsilon, \epsilon]} \Lambda_0\left(\lambda+\delta N \sqrt {N}\lambda\right)\geq \max_{\delta \in [-\epsilon', \epsilon']} \Lambda_0\left(\lambda+\delta N \sqrt {N}\lambda\right)=g_0(\epsilon',\lambda).\] Thus, for fixed $\lambda$, $f(\cdot,\lambda)$ is non-increasing, and the claim of the Lemma follows.
\end{proof}
We proceed by bounding further the right hand side in~\eqref{eq-Main-bound},
 by rewriting $e^{-(k-(s+1))(|\log r|-\delta )}$ as $\frac{1}{r e^\delta}\,e^{-(k-s)(|\log r|-\delta)}$:
\begin{eqnarray}
\label{eq-Main-bound2}
&\,& \sum_{s=0}^{k}\frac{C(\delta)}{r e^\delta}\, e^{sN g_0(\epsilon,\lambda)\,+\, (k-s)\Lambda_0(N\lambda)\,-\,(k-s)(|\log r|-\delta)}
\nonumber\\
&\leq&
(k+1) \, \max_{s\in \{0,\ldots,k\}} \frac{C(\delta)}{r e^\delta}
e^{ \left[ \,sN g_0(\epsilon,\lambda) \,+\, (k-s)\, \left( \Lambda_0(N\lambda)- (|\log r|-\delta) \right) \,\right] }\nonumber\\
&=&
(k+1) \, \frac{C(\delta)}{r e^\delta}  \,e^{ \max_{s\in \{0,\ldots,k\}}
\left[ \,s N g_0(\epsilon,\lambda)\,+\,(k-s) \left(  \Lambda_0(N\lambda)- (|\log r|-\delta) \right) \, \right]}\nonumber\\
&\leq &
(k+1) \, \frac{C(\delta)}{r e^\delta}  \,e^{ k \max_{\theta \in[0,1]}
\left[ \,
\theta N g_0(\epsilon,\lambda)\,+\,(1-\theta)\left(\Lambda_0(N\lambda)- (|\log r|-\delta)\right)
\right]}\nonumber\\
\label{eqn-corrected}
&=&  (k+1)\, \frac{C(\delta)}{re^\delta} \, e^{ k \, \left[ \,\left(N g_0(\epsilon,\lambda), \Lambda_0(N\lambda)- (|\log r|-\delta) \right)\,\right]}.
\end{eqnarray}
The second inequality follows by introducing $\theta:=\frac{s}{k}$ and by enlarging the set for $\theta$ from $\left\{0,\frac{1}{k},\ldots, 1\right\}$ to the continuous interval $[0,1]$.
Taking the $\log$ and dividing by $k$, from~\eqref{eq-Exp-Markov} and~\eqref{eqn-corrected} we get:
\begin{eqnarray}
\label{eq-Main-bound3}
\frac{1}{k}\log \alpha_i(k,\gamma)& \leq & \frac {\log (k+1)}{k}+\frac {\log \frac{C(\delta)}{r e^\delta}}{k} + \max \left\{N g_0(\epsilon,\lambda), \Lambda_0(N\lambda)- (|\log r|-\delta)\right\}  - N\gamma \lambda.
\end{eqnarray}
Taking the $\limsup$ when $k \rightarrow \infty$, the first two terms in the right hand side of~\eqref{eq-Main-bound3} vanish;
  further, changing the sign, we get a bound on the exponent of $\alpha_i(k)$ that holds for every $\epsilon>0$:
\begin{eqnarray*}
\label{eq-Main-bound4}
\liminf -\frac{1}{k}\log \alpha_i(k,\gamma)& \geq &  - \max \left\{ N g_0(\epsilon,\lambda), \, \Lambda_0(N\lambda)- (|\log r|-\delta) \right\} + N\gamma \lambda.
\end{eqnarray*}
By Lemma~\ref{lemma-inf-f-epsilon}, as $\epsilon \rightarrow 0$, $N g_0(\epsilon,\lambda)$ decreases to $N\,\Lambda_0(\lambda)$;
 further, letting $\delta \rightarrow 0$, we get
\begin{eqnarray}
\label{eq-Main-bound5}
\liminf -\frac{1}{k}\log \alpha_i(k,\gamma)& \geq &  - \max \left\{ N \Lambda_0(\lambda), \, \Lambda_0(N\lambda)- |\log r| \right\}
+N\gamma \lambda.
\end{eqnarray}
The previous bound on the exponent of the probability of false alarm holds for any $\lambda\geq 0$. To get the best bound, we maximize the expression on the right hand side of~\eqref{eq-Main-bound5} over $\lambda\in [0,\infty)$.
(We refer to figure \ref{Fig-corollary} to help illustrate the bounds $B_0(\gamma)$ and $B_1(\gamma)$
 for a discrete valued observations $Y_i(t)$ over a 5-point alphabet.) To this end, introduce
\begin{equation}
\label{eqn-Phi-0-def}
\Phi_0(\lambda):=\max \left\{ N \Lambda_0(\lambda), \, \Lambda_0(N\lambda)- |\log r| \right\}.
\end{equation}
We show that the best bound equals $B_0(\gamma)$ in~\eqref{theorem-rate-bound-for-alpha}, i.e.:
\begin{equation}
\label{eq-B0-as-conjugate}
B_0(\gamma)= \max_{\lambda \geq 0}N\gamma \lambda-\Phi_0(\lambda).
\end{equation}
From the first order optimality conditions, for a fixed $\gamma$, an optimizer $\lambda^\star=\lambda^\star(\gamma)$ (if it exists) of the objective in~\eqref{eq-B0-as-conjugate} is a point that satisfies:
\begin{equation}
\label{eq-first-order-opt-condition}
N\gamma \in \partial \Phi_0(\lambda^\star),\;\; \lambda^\star\geq 0,
 \end{equation}
 where $\partial \Phi_0(\lambda)$ denotes the subdifferential set of $\Phi_0$ at $\lambda$. We next characterize $\partial \Phi_0(\lambda)$, for $\lambda\geq 0$. Recall the zero $\lambda_0^{\mathrm{s}}$ of $\Delta_0(\cdot)$ from Theorem~\ref{Theorem-rate-bounds-for-alpha-and-beta}. The subdifferential $\partial \Phi_0(\lambda)$ is:
\begin{equation}
\label{eq-subdifferential-of-Phi_0}
\partial \Phi_0(\lambda)=\left\{ \begin{array}{lll} \{N\Lambda_0^\prime(\lambda)\}, & \mathrm{for} &
\lambda\in [0,\lambda_0^{\mathrm{s}}) \\
\left[N\Lambda_0^\prime(\lambda), N \Lambda_0^\prime(N\lambda)\right] , & \mathrm{for}&    \lambda= \lambda_0^{\mathrm{s}} \\
\{N\Lambda_0^\prime(N\lambda)\}, & \mathrm{for} &    \lambda> \lambda_0^{\mathrm{s}}. \end{array} \right.
\end{equation}
We next find $B_0(\gamma)$ for any $\gamma\in(\gamma_0,\gamma_1)$, by finding $\lambda^\star=\lambda^\star(\gamma)$ for any $\gamma\in(\gamma_0,\gamma_1)$. Recall $\gamma_0^{-}$ and $\gamma_0^{+}$ from Theorem~\ref{Theorem-rate-bounds-for-alpha-and-beta}. We separately consider three regions: 1) $\gamma\in [\gamma_0,\gamma_0^{-}]$; 2) $\gamma\in (\gamma_0^{-},\gamma_0^{+})$;
and 3) $\gamma\in [\gamma_0^{+},\gamma_1]$. For the first region, recall that $\Lambda_0^\prime(0)=\gamma_0$, i.e., for $\gamma=\gamma_0$, equation~\eqref{eq-first-order-opt-condition} holds (only) for $\lambda^\star=0$. Also, for $\gamma=\gamma_0^{-}$, we have $\Lambda_0^\prime(\lambda_0^{\mathrm{s}})=\gamma_0^{-}$, i.e., equation~\eqref{eq-first-order-opt-condition} holds (only) for $\lambda^\star=\lambda_0^{\mathrm{s}}$. Because $\Lambda_0^\prime(\lambda)$ is continuous and strictly increasing on $\lambda \in [0,\lambda_0^{\mathrm{s}}]$, it follows that, for any $\gamma \in [\gamma_0,\gamma_0^{-}]$ there exists a solution to~\eqref{eq-first-order-opt-condition}, it is unique, and lies in $[0,\lambda_0^{\mathrm{s}}]$. Now, we calculate $B_0(\gamma)$:
\begin{align}
\label{eq-calculate-B_0-region-1}
B_0(\gamma)&=N\lambda^\star \gamma-\Phi_0(\lambda^\star)= N\lambda^\star \gamma-N\Lambda_0(\lambda^\star) \\
&=N (\lambda^\star\gamma- \Lambda_0(\lambda^\star))=N \sup_{\lambda\geq 0} (\lambda \gamma- \Lambda_0(\lambda))=NI_0(\gamma),
\end{align}
where we used the fact that $\Phi_0(\lambda^\star)=N\Lambda_0(\lambda^\star)$ (because $\lambda^\star\leq \lambda_0^{\mathrm{s}}$),
 and the definition of the function $I_0(\cdot)$ in~\eqref{eqn-I-0-def}.
We now consider the second region. Fix $\gamma \in(\gamma_0^{-}, \gamma_0^{+})$. It is trivial to verify, from~\eqref{eq-subdifferential-of-Phi_0}, that $\lambda^\star=\lambda_0^{\mathrm{s}}$ is the solution to~\eqref{eq-first-order-opt-condition}. Thus, we calculate $B_0(\gamma)$ as follows:
\begin{align}
\label{eq-calculate-B_0-region-2}
B_0(\gamma)&=N\lambda_0^{\mathrm{s}} \gamma-\Phi_0(\lambda_0^{\mathrm{s}})=
N\lambda_0^{\mathrm{s}} \gamma-N\Lambda_0(\lambda_0^{\mathrm{s}}) \\
&=N\lambda_0^{\mathrm{s}} (\gamma-\gamma_0^{-}) + N\lambda_0^{\mathrm{s}} \gamma_0^{-} -N\Lambda_0(\lambda_0^{\mathrm{s}})=
N\lambda_0^{\mathrm{s}} (\gamma-\gamma_0^{-}) + N I_0(\gamma_0^{-}),
\end{align}
where we used the fact that $\lambda_0^{\mathrm{s}} \gamma_0^{-} -\Lambda_0(\lambda_0^{\mathrm{s}})=
 \sup_{\lambda\geq 0} \lambda \gamma_0^{-}- \Lambda_0(\lambda)=I_0(\gamma_0^{-})$. The proof for the third region is analogous to the proof for the first region.

For a proof of the claim on the probability of miss $\beta_i(k,\gamma)=\mathbb{P} \left(x_i(k)< \gamma|H_1\right)$, we proceed analogously to~\eqref{eq-Exp-Markov-inequality}, where instead of $\zeta\geq 0$, we now use $\zeta\leq 0$ (and, hence, the proof proceeds with $\lambda \leq 0$).
\end{proof}
%
%
%
%
%
%We now prove Corollary~\ref{Corollary-main-result}.
%
\begin{proof}[Proof of Corollary~\ref{Corollary-main-result}]
We first prove part (a). Consider the error probability
$P_{\mathrm{e},i}(k,\gamma)$ in~\eqref{eqn-P-fa}.
By Lemma 1.2.15 in \cite{DemboZeitouni}, we have that:
\begin{align*}
\liminf_{k \rightarrow \infty} -\frac{1}{k} \log P_{\mathrm{e},i}(k,\gamma)
 &= \min \left\{ \liminf_{k \rightarrow \infty} -\frac{1}{k} \log (\alpha_i(k,\gamma) \pi_0),
  \liminf_{k \rightarrow \infty} -\frac{1}{k} \log (\beta_i(k,\gamma)\pi_1) \right\}\\
  &=\min \left\{ \liminf_{k \rightarrow \infty} -\frac{1}{k} \log \alpha_i(k,\gamma) ,
  \liminf_{k \rightarrow \infty} -\frac{1}{k} \log \beta_i(k,\gamma) \right\}\\
  & \geq     \min\{B_0(\gamma),B_1(\gamma)\},
\end{align*}
where last inequality is by Theorem \ref{Theorem-rate-bounds-for-alpha-and-beta}.
 We now show
 that $\min\{B_0(\gamma),B_1(\gamma)\}>0$ for all $\gamma \in (\gamma_0,\gamma_1).$
 First, from the expression for $B_0(\gamma)$ in Theorem~\ref{Theorem-rate-bounds-for-alpha-and-beta},
  for $|\log r|>0$, we have:
  $B_0(\gamma_0)=N I_0(\gamma_0)=0$,
   and $B_0^\prime(\gamma) = N I_0^\prime(\gamma)>0$ for any
   $\gamma \in (\gamma_0,\gamma_0^{-})$.
   As the function $B_0(\cdot)$ is convex,
   we conclude that $B_0(\gamma)>0$, for all $\gamma>\gamma_0.$
  (The same conclusion holds under $|\log r|=0,$
   by replacing $N I_0(\gamma)$ with $I_0(\gamma)+|\log r| = I_0(\gamma).$)
    Analogously, it can be shown that $B_1(\gamma)>0$ for all $\gamma <\gamma_1$,
     and so $\min\{B_0(\gamma),B_1(\gamma)\}>0$, for all $\gamma \in (\gamma_0,\gamma_1).$

 We now calculate $\max_{\gamma \in (\gamma_0,\gamma_1)}
  \min\{B_0(\gamma),B_1(\gamma)\}$. Consider the function $\Delta_B(\gamma):=B_1(\gamma)-B_0(\gamma)$.
   Using the definition of $B_0(\gamma)$
    in Theorem \ref{Theorem-rate-bounds-for-alpha-and-beta}, and taking the subdifferential
     of $B_0(\gamma)$ at any point $\gamma \in (\gamma_0,\gamma_1)$,
     it is easy to show that $B_0^\prime(\gamma)>0$, for any
     subgradient $B_0^\prime(\gamma) \in \partial B_0(\gamma)$,
     which implies that $B_0(\cdot)$ is strictly
     increasing on $\gamma \in (\gamma_0,\gamma_1)$.
     Similarly, it can be shown that $B_1(\cdot)$
      is strictly decreasing on $\gamma \in (\gamma_0,\gamma_1)$. Further, using the properties
        that $I_0(\gamma_0)=0$ and $I_1(\gamma_1)=0$, we have
        $\Delta_B(\gamma_0) = B_1(\gamma_0)>0$, and
        $\Delta_B(\gamma_1) = -B_0(\gamma_1)<0$.
       By the previous two observations, we have that $\Delta_B(\gamma)$ is strictly decreasing on $\gamma \in (\gamma_0,\gamma_1)$, with $\Delta_B(\gamma_0)>0$ and $\Delta_B(\gamma_1)<0$. Thus,
       $\Delta_B(\cdot)$ has a unique zero $\gamma^\star$ in $\gamma \in (\gamma_0,\gamma_1)$.
       Now, the fact that
         $\max_{\gamma \in (\gamma_0,\gamma_1)}
  \min\{B_0(\gamma),B_1(\gamma)\} = B_0(\gamma^\star)=B_1(\gamma^\star)$
   holds trivially because $B_0(\cdot)$ is
   strictly increasing on $\gamma \in (\gamma_0,\gamma_1)$
    and $B_1(\cdot)$ is
   strictly decreasing on $\gamma \in (\gamma_0,\gamma_1)$. This
   completes the proof of part (a).

   We now prove part~(b).
   Suppose that $|\log r| \geq \mathrm{thr}(\Lambda_0,N).$
   We show that, for $\gamma=0$:
   \begin{eqnarray}
   \label{eqn-B0-equal_I0}
   B_0(0) = N I_0(0),\:\:\:\:\:
   %\label{eqn-B1-equal_I1}
   B_1(0) = N I_1(0) = N I_0(0).
   \end{eqnarray}
   (Last equality in \eqref{eqn-B0-equal_I0} holds because
   $I_1(0)=(I_0(\gamma)-\gamma)|_{\gamma=0}=I_0(0)$.)
   Equations \eqref{eqn-B0-equal_I0}
    mean that $B_0(0)=B_1(0)$. Further, $0 \in (\gamma_0,\gamma_1)$,
    and, from part (a), $\gamma^\star$ is unique, and so
    $\gamma^\star$ has to be $0$. This shows
    that $\sup_{\gamma \in (\gamma_0,\gamma_1)}\min\{B_0(\gamma),B_1(\gamma)\}
    = N I_0(0) = N C_{\mathrm{ind}}$, and so, by part~(a):
    \begin{equation}
    \label{eqn-liminf-novo}
    \liminf_{k \rightarrow \infty} -\frac{1}{k} \log P_{\mathrm{e},i}(k,0) \geq N C_{\mathrm{ind}}.
    \end{equation}
On the other hand,
\begin{equation}
\label{eqn-limsup}
\limsup_{k \rightarrow \infty} -\frac{1}{k} \log P_{\mathrm{e,i}}(k,0) \leq N C_{\mathrm{ind}},
\end{equation}
because, by the Chernoff lemma~\cite{DemboZeitouni}, for \emph{any test} (with the corresponding
error probability $P_e^\prime(k,\gamma)$,) we have that $\limsup_{k \rightarrow \infty}-\frac{1}{k} \log P_e^\prime(k,\gamma) \leq N C_{\mathrm{ind}}$. Combining~\eqref{eqn-liminf-novo} and~\eqref{eqn-limsup} yields`
\[
N C_{\mathrm{ind}} \leq \liminf_{k \rightarrow \infty} -\frac{1}{k} \log P_{\mathrm{e,i}}(k,0)
\leq \limsup_{k \rightarrow \infty} -\frac{1}{k} \log P_{\mathrm{e,i}}(k,0) \leq N C_{\mathrm{ind}}.
\]
To complete the proof of part~(b),
it remains to show \eqref{eqn-B0-equal_I0}. We
prove only equality for $B_0$ as equality for $B_1$ follows similarly.
Because $|\log r| \geq \mathrm{thr}(\Lambda_0,N)$, we have,
from the definition of $\Phi_0(\cdot)$ in~\eqref{eqn-Phi-0-def}, that
$\Phi_0(\lambda^\bullet) = N \Lambda_0(\lambda^\bullet).$
Recall that $B_0(0) = -\Phi_0(\lambda^\star)$,
where $\lambda^\star$ is
a point for which~\eqref{eq-first-order-opt-condition} holds for $\gamma=0$.
However, because $\partial \Phi_0(\lambda^\bullet) = \{N \Lambda_0^\prime(\lambda^\bullet)\}$,
and $\Lambda_0^\prime(\lambda^\bullet)=0$, it follows that $\lambda^\star = \lambda^\bullet$
 and $B_0(0)=-\Phi_0(\lambda^\bullet) = - N \Lambda_0(\lambda^\bullet) = N I_0(0)$,
 which proves \eqref{eqn-B0-equal_I0}. Thus, the result in part (b) of the Lemma.
%\end{proof}
%

\end{proof}
\vspace{-4mm}
\section{Examples}
\label{sec-Examples}
This section illustrates our main results for several examples of the distributions of the sensor
observations. Subsection \ref{subsec-Gaussian-example} compares the
Gaussian and Laplace distributions, both with
a finite number of sensors $N$ and when $N \rightarrow \infty$. Subsection \ref{subsec-Discrete-example} considers discrete distributions with finite support, and, in more detail, binary distributions.
Finally, Subsection \ref{subsection-topology} numerically demonstrates that
our theoretical lower bound on the error exponent \eqref{eqn-rate} is tight. Subsection \ref{subsection-topology} also shows
trhough a symmetric, tractable example how distributed detection performance depends on the network topology (nodes' degree and
link occurrence/failure probability.)
%
%
%In this section we illustrate our main result with several examples of sensor measurement distribution.
%
\subsection{Gaussian distribution versus Laplace distribution}
\label{subsec-Gaussian-example}

\textbf{Gaussian distribution.} We now study the detection of a signal in additive Gaussian noise; $Y_i(t)$ has the following density:
\begin{equation*}
f_{\mathrm{G}}(y)=\,\left\{ \begin{array}{ll}
\frac{1}{\sqrt{2 \pi}\sigma_{\mathrm{G}}}e^{-\frac{\left(y-m_{\mathrm{G}}\right)^2}{2\sigma_{\mathrm G}^2}}, & H_1\\  \frac{1}{\sqrt{2 \pi}\sigma_{\mathrm{G}}}e^{-\frac{y^2}{2\sigma_{\mathrm G}^2}}, & H_0. \end{array} \right.
\end{equation*}
The LMGF is given by:
$\Lambda_{0,\mathrm G}(\lambda)=-\frac{\lambda(1-\lambda)}{2} \frac{m_{\mathrm{G}}^2}{\sigma_{\mathrm{G}}^2}.
$
The minimum of $\Lambda_{0,\mathrm G}$ is achieved at $\lambda^\bullet=\frac{1}{2}$, and the per sensor Chernoff information is $C_{\mathrm{ind},\mathrm{G}}=\frac{m_{\mathrm{G}}^2}{8 \sigma_{\mathrm{G}}^2}.$

Applying Corollary~\ref{Corollary-main-result}, we get the sufficient condition for optimality:
\begin{equation}
\label{eqn-threshold-gauss}
|\log r|\geq \Lambda_{0,\mathrm G}\left(\frac{N}{2}\right)-N\Lambda_{0,\mathrm G}\left(\frac{1}{2}\right)=N(N-1)C_{\mathrm{ind},\mathrm{G}}.
\end{equation}

Since $\Lambda_0(\lambda)=\Lambda_1(\lambda)$, the two conditions from the Corollary here reduce to a single condition in \eqref{eqn-rate}.

Now,  let the number of sensors $N \rightarrow \infty$, while keeping the total Chernoff information constant, i.e., not dependent on $N$; that is,
            $C_{\mathrm{G}}: = N C_{\mathrm{ind,G}} = \mathrm{const}$, $C_{\mathrm{ind,G}}(N) = C_{\mathrm{G}}/N.$  Intuitively, as $N$ increases, we deploy more and more sensors over a region (denser deployment),
             but, on the other hand, the sensors' quality becomes worse and worse.
             The increase of $N$ is balanced in such a way that
             the total information offered by all sensors stays constant with $N$.
             Our goal is to determine how the optimality threshold
             on the network connectivity $\mathrm{thr}(N,\Lambda_{0,\mathrm{G}})$ depends on $N$. We can
             see from~\eqref{eqn-threshold-gauss} that
               the optimality threshold for the distributed detector in the Gaussian case equals:
           \begin{equation}
           \label{eqn-thr-gaussian}
           \mathrm{thr}(\Lambda_{0,\mathrm{G}},N)=(N-1)C_{\mathrm{G}}.
           \end{equation}

\mypar{Laplace distribution} We next study the optimality conditions for the sensor observations with Laplace distribution. The density of $Y_i(t)$ is:
\begin{equation*}
f_{\mathrm{L}}(y)=\,\left\{ \begin{array}{lll}  \frac{1}{2 b_{\mathrm{L}}}e^{-\frac{|y-m_{\mathrm{L}}|}{b_{\mathrm L}}}, & H_1\\  \frac{1}{2 b_{\mathrm{L}}}e^{-\frac{|y|}{b_{\mathrm L}}}, & H_0. \end{array} \right.
\end{equation*}
The LMGF has the following form:
\begin{equation*}
\Lambda_{0,\mathrm L}(\lambda)=\log \left(\frac{1-\lambda}{1-2\lambda} e^{-\lambda  \frac{m_{\mathrm{L}}}{ b_{\mathrm{L}}}} -  \frac{\lambda}{1-2\lambda}e^{-\left(1-\lambda\right)  \frac{m_{\mathrm{L}}}{b_{\mathrm{L}}}}\right).
\end{equation*}
Again, the minimum is at $\lambda^\bullet=\frac{1}{2}$, and the per sensor Chernoff information is
\[ C_{\mathrm{ind},\mathrm{L}}=\frac{m_{\mathrm{L}}}{2 b_{\mathrm{L}}} - \log\left(1+\frac{m_{\mathrm{L}}}{2 b_{\mathrm{L}}}\right).\]
The optimality condition in~\eqref{eqn-rate} becomes:
\begin{eqnarray}
\label{eqn-thr-laplace}
|\log r| & \geq & \Lambda_{0,\mathrm{L}}\left( \frac{N}{2} \right) - N \Lambda_{0,\mathrm{L}}\left( \frac{1}{2}\right)
\\
&=& \log \left(
  \frac{2-N}{2-2N} e^{-\frac{N}{2}\frac{m_{\mathrm{L}}} {b_{\mathrm{L}}}} -  \frac{N}{2-2N}e^{-(1-\frac{N}{2})
\frac{ m_{\mathrm{L} }} {b_{\mathrm{L}}}}\right)-N \log \left(1+\frac{m_{\mathrm{L}}}{2b_{\mathrm{L}}} \right) + N\frac{m_{\mathrm{L}}}{2 b_{\mathrm{L}}}.
\nonumber
\end{eqnarray}

\mypar{Gaussian versus Laplace distribution} It is now interesting to compare the Gaussian and the Laplace
case under equal per sensor Chernoff information $C_{\mathrm{ind,L}} = C_{\mathrm{ind,G}}$.
 Figure~\ref{fig-2} (left)
  plots the LMGF for the Gaussian and Laplace distributions,
   for $N=10$, $C_{\mathrm{ind}} = C_{\mathrm{ind,L}} = C_{\mathrm{ind,G}} = 0.0945$,
    $b_{\mathrm{L}}=1$, $m_{\mathrm{L}}=1,$ and $m_{\mathrm{G}}^2/\sigma_{\mathrm{G}}^2 =0.7563 = 8 C_{\mathrm{ind}}.$
     By~\eqref{eqn-threshold}, the optimality
      threshold equals \[|N \Lambda_0(1/2)|+|\Lambda_0(N/2)|,\]
      as $\lambda^\bullet = 1/2$, for both
      the Gaussian and Laplace distributions.
      The threshold can be estimated from Figure~\ref{fig-2} (left): solid
      lines plot the functions $\Lambda_0(N \lambda)$ for
      the two different distributions, while
      dashed lines plot the functions $N\,\Lambda_0(\lambda)$.
      For both solid and dashed lines, the Gaussian
      distribution corresponds to the more curved functions. We see that the threshold is larger for the Gaussian case. This means that,
       for a certain range $r \in (r_{\mathrm{min}}, r_{\mathrm{max}})$,
        the distributed detector
        with Laplace sensors is asymptotically optimal,
         while with Gaussian sensors the distributed detector
         may not be optimal, even though it uses the same network infrastructure (equal $r$) and has equal per sensor Chernoff information. (See also Figure \ref{fig-2} (right) for another illustration of this effect.)

         We now compare the Gaussian and Laplace distributions when $N \rightarrow \infty$,
          and we keep the Gaussian total Chernoff
          information $C_{\mathrm{G}}$ constant with $N$. Let the Laplace distribution parameters vary with $N$ as:
                \begin{eqnarray*}
         m_{\mathrm{L}} = m_{\mathrm{L}}(N) = \frac{2 \sqrt{2 C_{\mathrm{G}}}}{\sqrt{N}},\:\:\:\:\:
         b_{\mathrm{L}} = b_{\mathrm{L}}(N)=1.
         \end{eqnarray*}
         We can show that, as $N \rightarrow \infty$,
          the total Chernoff information $C_{\mathrm{L}}(N) \rightarrow C_{\mathrm{G}}$ as
          $N \rightarrow \infty$, and so the Gaussian and the Laplace
          centralized detectors become equivalent.
          On the other hand,
          the threshold for the Gaussian distributed detector is given by~\eqref{eqn-thr-gaussian}
  %        .
%          ,\footnote{As explained before,
%           the threshold in~\eqref{eqn-thr-gaussian} is not only sufficient
%           for optimality, but it is also necessary, for a specific
%           type of the random matrices $W(k)$ (the switching fusion type)~\cite{GaussianDD}.}
           while,
          for the Laplace detector, using \eqref{eqn-thr-laplace} and
          a Taylor expansion, we get that the optimality threshold is approximately:
          \[
          \mathrm{thr}(\Lambda_{0,\mathrm{L}},N) \approx \sqrt{2 C_{\mathrm{G}} N}.
          \]
          Hence, the required $|\log r|$ to achieve the optimal error exponent
           grows much slower with the Laplace distribution than with the Gaussian distribution.

\begin{figure}[thpb]
      \centering
      \includegraphics[height=2.1in,width=3.1 in]{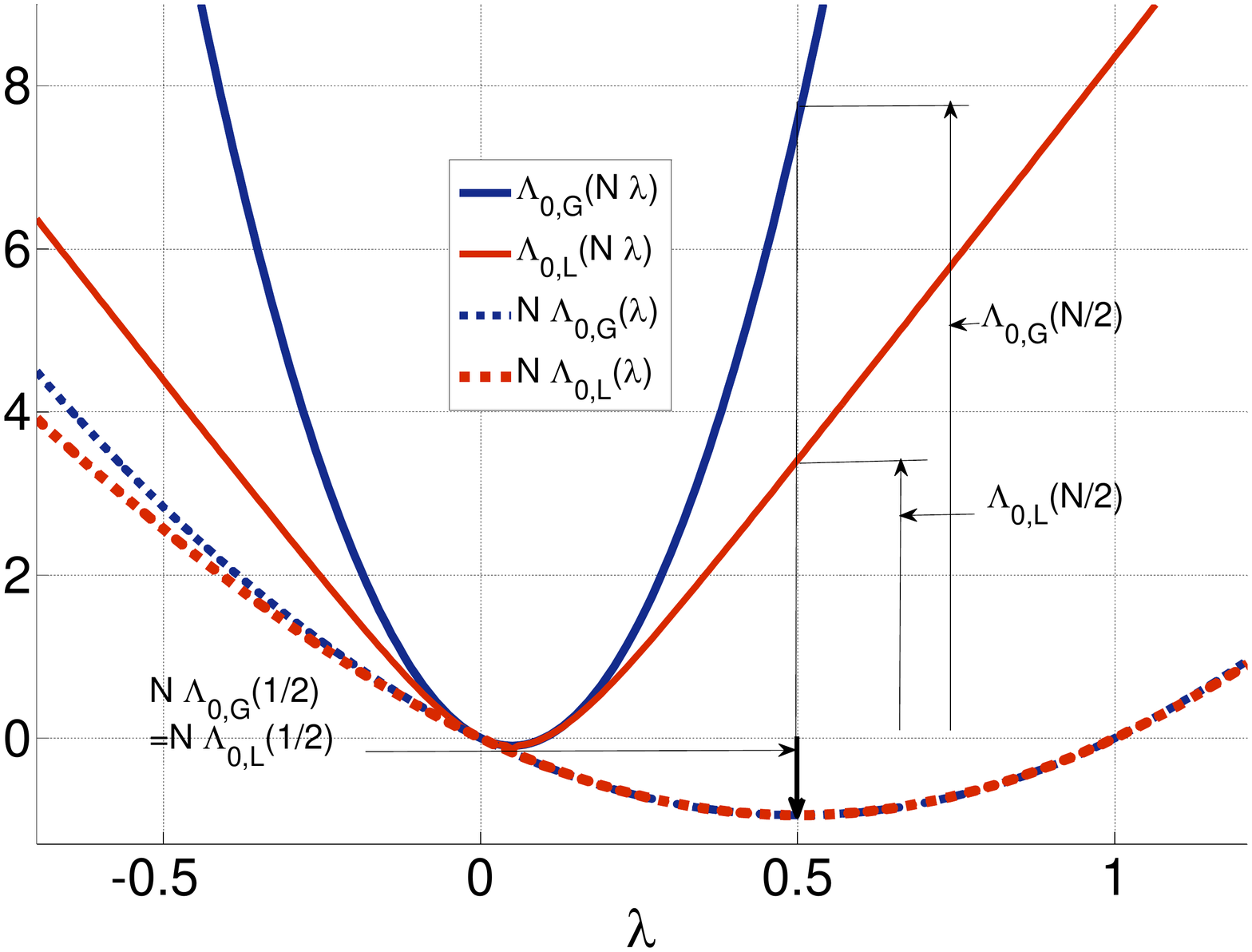}
      \includegraphics[height=2.1in,width=3.1 in]{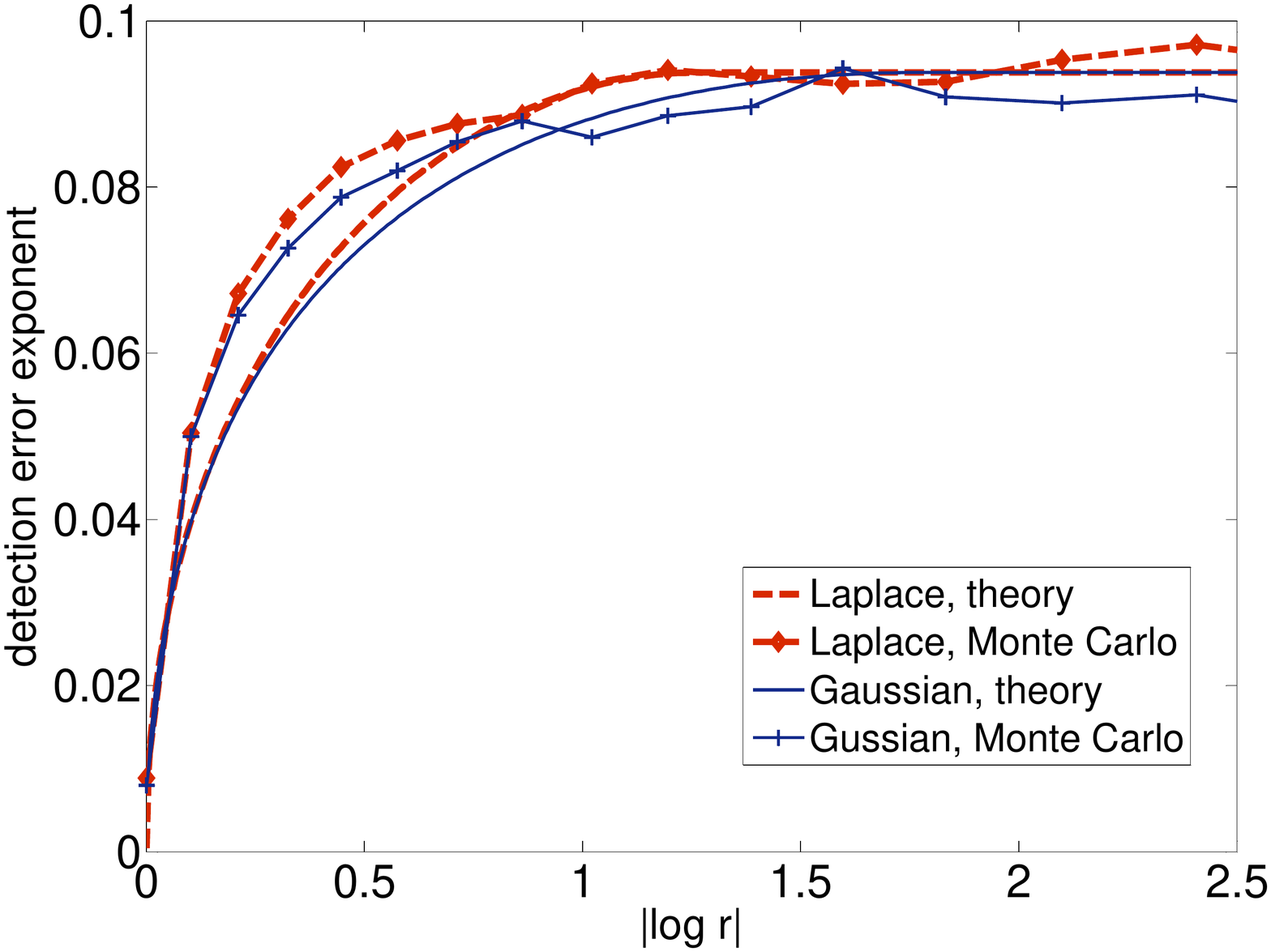}
      \caption{\textbf{Left:} LMGFs for Gaussian and Laplace distributions with equal per sensor Chernoff informations,
   for $N=10$, $C_{\mathrm{ind}} = C_{\mathrm{ind,L}} = C_{\mathrm{ind,G}} = 0.0945$,
    $b_{\mathrm{L}}=1$, $m_{\mathrm{L}}=1,$ and $m_{\mathrm{G}}^2/\sigma_{\mathrm{G}}^2 =0.7563 = 8 C_{\mathrm{ind}}. $
    Solid lines plot the functions $\Lambda_0(N\lambda)$ for the two distributions,
     while dashed lines plot the functions $N \Lambda_0(\lambda)$. For both solid and dashed lines, the Gaussian
      distribution corresponds to the more curved functions. The optimality threshold in~\eqref{eqn-threshold}
      is given by $|N \Lambda_0(1/2)|+|\Lambda_0(N/2)|$, as $\lambda^\bullet=1/2$. \textbf{Right:} Lower bound on the error exponent in~\eqref{eqn-rate} and Monte Carlo estimate of the error exponent versus $|\log r|$ for the Gaussian and Laplace sensor
      observations: $N=20$, $C_{\mathrm{ind}} = C_{\mathrm{ind,L}} = C_{\mathrm{ind,G}} = 0.005$,
    $b_{\mathrm{L}}=1$, $m_{\mathrm{L}}=0.2,$ and $m_{\mathrm{G}}^2/\sigma_{\mathrm{G}}^2 =
      0.04 = 8 C_{\mathrm{ind}}.$}
      \label{fig-2}
\end{figure}
\subsection{Discrete distributions}
\label{subsec-Discrete-example}
We now consider the case when the support of the sensor observations
under both hypothesis is a finite alphabet
$\{a_1,a_2,...,a_M\}$. This case is of  practical interest
 when, for example, the sensing device has
 an analog-to-digital converter with a finite range;
 hence, the observations take only a
 finite number of values. Specifically, the distribution of
 $Y_i(k)$, $\forall i$, $\forall k$, is given by:
 \begin{equation}
 \label{eqn-discrete-dis}
 \mathbb{P}(Y_i(k)=a_m) = \left\{ \begin{array}{ll}   q_m, & H_1\\  p_m, & H_0 \end{array} \right.,\:\:m=1,...,M.
 \end{equation}
Then, the LMGF under $H_0$ equals:
\[
\Lambda_0(\lambda) = \log \left( \sum_{m=1}^M q_m^\lambda p_m^{1-\lambda} \right).
\]
Note that $\Lambda_0(\lambda)$ is finite on ${\mathbb R}$.
 Due to concavity of $-\Lambda_0(\cdot)$, the argument of the Chernoff information $\lambda^\bullet$
 ($C_{\mathrm{ind}} = \max_{\lambda_{\in [0,1]}} \{-\Lambda_0(\lambda)\} = - \Lambda_0(\lambda^\bullet)$)
  can, in general, be efficiently computed numerically, for example, by the Netwon method (see, e.g.,~~\cite{BoydsBook},
   for details on the Newton method.) It can be shown, defining $c_m = \log\left( \frac{q_m}{p_m}\right)$, that the Newton direction, e.g.,~\cite{BoydsBook} equals:
   \begin{eqnarray*}
   d(\lambda) = -\left(\Lambda_0^{\prime \prime}(\lambda)\right)^{-1} \Lambda_0^{\prime}(\lambda)
   = -\frac{1}
   {\frac{\sum_{m=1}^M c_m^2 p_m e^{\lambda c_m} }{\sum_{m=1}^M c_m p_m e^{\lambda c_m} }-\frac{\sum_{m=1}^M c_m p_m e^{\lambda c_m} }
  { \sum_{m=1}^M c_m e^{\lambda c_m}}}. \nonumber
   \end{eqnarray*}

   \mypar{Binary observations} To gain more intuition
   and obtain analytical results, we consider~\eqref{eqn-discrete-dis} with $M=2$, i.e., binary sensors,
%    \[
%     \mathbb{P}(Y_i(k)=\gamma_m) = \left\{ \begin{array}{ll}   q_m, & H_1\\  p_m, & H_0 \end{array} \right.,\:\:m=1,2,
%    \]
with $p_2=1-p_1 = 1-p$, $q_2=1-q_1=1-q$. Suppose further that $p<q.$
     We can show that the negative of the per sensor Chernoff information
     $\Lambda_{0,\mathrm{bin}}$ and the quantity $\lambda^\bullet$ are:
     \begin{eqnarray*}
     \label{eqn-chernoff-binary}
     -C_{\mathrm{ind}} &=& \Lambda_{0,\mathrm{bin}}(\lambda^\bullet)
     = \lambda^\bullet \log \left( \frac{q}{p} \right) + \log p
      + \log \left( 1 -  \frac{\log \frac{q}{p}}{\log\frac{1-q}{1-p}} \right)  \\
      \lambda^\bullet &=& \frac{ \log \frac{p}{1-p} + \log \left(\frac{\log \frac {p}{q}} {\log \frac{1-q}{1-p}}\right)   }
      {  \log \left( \frac{1-q}{1-p}\right) - \log \left(\frac{q}{p}\right)}.
     \end{eqnarray*}
Further, note that:
\begin{eqnarray}
\label{eqn-pom}
\Lambda_{0,\mathrm{bin}}(N \lambda^\bullet) &=& \log \left( p \left( \frac{q}{p}\right)^{N \lambda^\bullet} + (1-p) \left( \frac{1-q}{1-p}\right)^{N \lambda^\bullet}\right)
\leq \log \left( \frac{q}{p}\right)^{N \lambda^\bullet}
= N \lambda^\bullet \log\left( \frac{q}{p}\right). %\nonumber
\end{eqnarray}
Also, we can show similarly that:
\begin{eqnarray}
\label{eqn-pom2}
\Lambda_{0,\mathrm{bin}}(1-N(1- \lambda^\bullet)) \leq N(1-\lambda^\bullet) \log \left( \frac{1-p}{1-q}\right).
\end{eqnarray}
Combining \eqref{eqn-pom} and \eqref{eqn-pom2}, and applying Corollary \ref{Corollary-main-result} (equation~\eqref{eqn-rate}),
 we get that a sufficient condition for asymptotic optimality is:
 \begin{eqnarray*}
 |\log r|
  \geq \max \left\{
 N \log\frac{1}{p}- N \log \left( 1+ \frac{\left| \log \frac{q}{p} \right|}{\left| \log\frac{1-q}{1-p}\right|} \right),
 N \log \frac{1}{1-q} - N \log \left( 1+ \frac{\left| \log\frac{1-q}{1-p}\right|} {\left| \log \frac{q}{p} \right|}\right)
 \right\}.
 \nonumber
 \end{eqnarray*}
We further assume a very simplified
sufficient condition for optimality:
\begin{equation}
\label{eqn-simplified}
|\log r| \geq N \max \left\{ \left| \log p \right|,  \left| \log (1-q) \right| \right\}.
\end{equation}
The expression in~\eqref{eqn-simplified} is intuitive.
Consider, for example, the case $p=1/2$, so that
the right hand side in~\eqref{eqn-simplified} simplifies to: $N |\log(1-q)|$.
 Let $q$ vary from $1/2$ to $1$.
 Then, as $q$ increases,
 the per sensor Chernoff information increases, and
 the optimal centralized detector has better and better performance
 (error exponent.) That is,
 the centralized detector has a very low error probability
  after a very short observation interval $k$. Hence, for larger $q$, the
  distributed detector needs more connectivity to be able to ``catch up'' with the
  performance of the centralized detector. We compare numerically
  Gaussian and binary distributed detectors with equal per sensor Chernoff information,
   for $N=32$ sensors, $C_{\mathrm{ind}} = 5.11 \cdot 10^{-4}$, $m_G^2/\sigma^2_G = 8 C_{\mathrm{ind}}$,
   $p=0.1$, and $q=0.12$. Binary detector requires more connectivity
   to achieve asymptotic optimality ($r \approx 0.25$), while
   Gaussian detector requires $r \approx 0.5.$

\subsection{Tightness of the error exponent lower bound in~\eqref{eqn-rate} and impact of the network topology}
\label{subsection-topology}
\mypar{Assessment of the tightness of the error exponent lower bound in \eqref{eqn-rate}} We note that the result in~\eqref{eqn-rate} is a theoretical lower bound on the error exponent. In particular, the condition $|\log r| \geq \mathrm{thr}(\Lambda_0,N)$ is proved to be a sufficient, but not necessary, condition
for asymptotically optimal detection; in other words,~\eqref{eqn-rate} does not exclude
the possibility of achieving asymptotic optimality for $|\log r|$ smaller than $\mathrm{thr}(\Lambda_0,N)$. In order
to assess the tightness of~\eqref{eqn-rate} (for both the Gaussian and Laplace distributions,) we perform Monte Carlo simulations to estimate the actual error exponent and compare it with~\eqref{eqn-rate}.
We consider $N=20$ sensors and fix the sensor observation distributions with the following parameters: $C_{\mathrm{ind}} = C_{\mathrm{ind,L}} = C_{\mathrm{ind,G}} = 0.005$,
    $b_{\mathrm{L}}=1$, $m_{\mathrm{L}}=0.2,$ and $m_{\mathrm{G}}^2/\sigma_{\mathrm{G}}^2 =0.04 = 8 C_{\mathrm{ind}}.$
We vary $r$ as follows. We construct a
(fixed) geometric graph with $N$ sensors by placing the nodes uniformly at random
on a unit square and connecting the nodes whose distance is less than a radius. Each link is a Bernoulli random variable, equal to $1$ with probability $p$ (link online),
and equal to $0$ with probability $1-p$ (link offline). The link occurrences are independent in time and space.
We change $r$ by varying $p$ from $0$ to $0.95$ in increments of $0.05$.
We adopt the Metropolis weights: whenever a link $\{i,j\}$ is online, we set $W_{ij}(k)=1/(1+\max(d_i(k),d_j(k)))$,
 where $d_i(k)$ is the number od neighbors of node $i$ at time $k$;
 when a link $\{i,j\}$ is offline, $W_{ij}(k)=0$; and $W_{ii}(k) = 1-\sum_{j \in O_i}W_{ij}(k)$,
 where we recall that $O_i$ is the neighborhood of node $i$. We obtain an estimate of the error probability $\widehat{P}_{\mathrm{e},i}(k)$ at sensor $i$ and time $k$ using 30,000 Monte Carlo runs of~\eqref{eqn-running-cons-sensor-i}
 per each hypothesis. We then estimate the sensor-wide average error exponent as:
 \[
 \frac{1}{N}\sum_{i=1}^N \frac{\log \widehat{P}_{\mathrm{e},i}(K_1) - \log \widehat{P}_{\mathrm{e},i}(K_2)}{K_2-K_1},
 \]
 with $K_1=40,$ $K_2=60.$ That is, we estimate the error exponent as the
 average slope (across sensors) of the error probability curve in a semi-log scale. Figure \ref{fig-2} (right) plots both the theoretical lower bound
 on the error exponent in~\eqref{eqn-rate} and the Monte Carlo estimate
 of the error exponent versus $|\log r|$ for Gaussian and Laplace distributions.
 We can see that the bound~\eqref{eqn-rate} is tight for both distributions.
 Hence, the actual distributed detection performance is very close to the performance
 predicted by~\eqref{eqn-rate}. (Of course, above the optimality threshold,~\eqref{eqn-rate} and
 the actual error exponent coincide and are equal to the total Chernoff information.)
 Also, we can see that the theoretical threshold on optimality $\mathrm{thr}(\Lambda_0,N)$
  and the threshold value computed from simulation are very close. Finally,
the distributed detector with Laplace observations achieves asymptotic optimality
  for a smaller value of $|\log r|$ ($|\log r| \approx 1.2$) than the distributed detector with Gaussian observations ($|\log r| \approx 1.6$),
  even though the corresponding centralized detectors are asymptotically equivalent.

\mypar{Impact of the network topology} We have seen in the previous two subsections
how detection performance depends on $r$. In order to understand how $r$ depends on
 the network topology, we consider a symmetric network structure, namely a regular network. For this case, we can express $r$
   as an explicit (closed form) function of the nodes' degrees and the link occurrence probabilities. (Recall that the smaller $r$ is, the better the network connectivity.)

Consider a connected regular network with $N$ nodes and degree $d \geq 2$.
Suppose that each link is a Bernoulli random variable, equal to $1$ with probability $p$ (link online)
and $0$ with probability $1-p$ (link offline,) with spatio-temporally independent link occurrences.
Then, it can be shown~\cite{rate-of-consensus} that $r$ equals:
\begin{equation}
\label{eqn-speed-cons-reg}
r = (1-p)^d.
\end{equation}
This expression is very intuitive. When $p$ increases,
i.e., when the links are online more often, the network (on average) becomes
more connected, and hence we expect that the network connectivity $|\log r|$ increases (improves). This is
confirmed by~\eqref{eqn-speed-cons-reg}: when $p$ increases, $r$ becomes smaller and closer to zero. Further, when $d$ increases, the network becomes more connected,
and hence the network speed again improves. Note also that
 $|\log r| = d |\log(1-p)|$ is a linear function of $d$.

We now recall Corollary~\ref{Corollary-main-result} to relate distributed detection performance
with $p$ and $d$. For example, for a fixed $p$, the distributed detection optimality condition becomes $d > \frac{\mathrm{thr}(\Lambda_0,N)}{|\log(1-p)|},$ i.e., distributed detection is asymptotically optimal when the sensors' degree is above a threshold. Further, because $d \leq N$, it follows that, for a large value of $\mathrm{thr}(\Lambda_0,N)$
and a small $p$, even networks with a very large degree (say, $d=N$) do not achieve asymptotic optimality.
Intuitively, a large $\mathrm{thr}(\Lambda_0,N)$ means that
the corresponding centralized detector decreases the error probability so fast in $k$ that,
because of the intermittent link failures, the distributed detector cannot ``catch up''
 with the centralized detector. Finally, when $p=1$,
 the optimality condition becomes $d > 0$, i.e.,
 distributed detection is asymptotically optimal for any $d \geq 2$.
 This is because, when $p=1$, the network is always connected,
 and the distributed detector asymptotically ``catches up''
 with the arbitrarily fast centralized detector. In fact,
 it can be shown that an arbitrarily connected network with no link failures achieves asymptotic
 optimality for any value of $\mathrm{thr}(\Lambda_0, N).$ (It can be shown
 that such a network has $r=0$, and, consequently, the network connectivity $|\log r|$ is $\infty$.)

%
%Recall that distributed detection performance has extremely nonlinear (phase transition type)
% behavior with respect to $r$: there is a sufficiently large value of $|\log r^\star| = \mathrm{thr}(\Lambda_0,N)$, above which
% further increase does not improve distributed detection performance. This behavior
%  translates also into intuitive network parameters, $p$ and $d$.
%  For the regular network example considered here,

%
%
%Note that, when $p=1$, there are no link
%failures, and the network is always connected. In this case,
%$r=0$, and distributed detector achieves asymptotic optimality for
%any value of $\mathrm{thr}(\Lambda_0,N).$
%

\section{Non-identically distributed observations}
\label{section-extensions}
We extend Theorem~\ref{Theorem-rate-bounds-for-alpha-and-beta} and Corollary~\ref{Corollary-main-result} to the case of (independent) \textit{non-identically} distributed observations. First, we briefly explain the measurement model and define the relevant quantities.
   As before, let $Y_i(t)$ denote the observation of sensor $i$ at time $t$, $i=1,\ldots,N$, $t=1,2,\ldots$.
%\setcounter{theorem}{5}
%\begin{assumption}

\emph{Assumption~A}
%\label{ass_1-non-id}
The observations of sensor $i$ are i.i.d. in time, with the following distribution:
\begin{equation*}
Y_i(t)\,\sim\,\left\{ \begin{array}{lr}   \nu_{i,1}, \;\;\;\; H_1 \\  \nu_{i,0}, \;\;\;\; H_0 \end{array} \right.,\:\:i=1,...,N,\:t=1,2,...
\end{equation*}
(Here we assume that $\nu_{i,1}$ and $\nu_{i,0}$ are mutually absolutely continuous, distinguishable measures, for $i=1,\ldots,N$).
Further, the observations of different sensors are independent both in time and in space, i.e., for $i\neq j$, $Y_i(t)$ and $Y_j(k)$ are independent for all $t$ and $k$.
%\end{assumption}

 Under Assumption~A, the form of the log-likelihood ratio test remains the same as under Assumption~\ref{ass_1}:
 \begin{equation*}\label{eq-centralized-dec-rule-2}D(k):=\frac{1}{Nk}\sum_{t=1}^k \sum_{i=1}^N   L_i(t)  \stackrel[H_0]{H_1}{\gtrless}\gamma,\end{equation*}
where the log-likelihood ratio at sensor $i$, $i=1,...,N$, is now:
\[L_i(t)=\log
\frac{f_{i,1} (Y_i(t))}{f_{i,0}(Y_i(t))},\]
where $f_{i,l}$, $l=0,1,$ is the density (or the probability mass) function associated with $\nu_{i,l}$. We now discuss the choice of detector thresholds $\gamma$. Let $\overline{\gamma}_l=\mathbb E\left[ \frac{1}{N}\sum_{i=1}^N L_i(t)|H_l\right]=\left(\sum_{i=1}^N\gamma_{i,l}\right)/{N}$. We will show that,
if $|\log r|>0$, any $\gamma \in (\overline{\gamma}_0, \overline {\gamma}_1)$
yields an exponentially fast decay of the error probability, at any sensor.
The condition $|\log r|>0$ means that the network is connected on average, e.g.,~\cite{SoummyaConferenceConsensus}; if met,
 then, for all $i$, $\mathbb{E}[x_i(k)|H_l] \rightarrow \overline{\gamma}_l$
  as $k \rightarrow \infty$, $l=0,1.$ (Proof is omitted for brevity.)
  Clearly, under identical sensors, $\gamma_{i,l}=\gamma_{j,l}$ for any $i,j$,
   and hence the range of detector thresholds becomes the one
   assumed in Section~\ref{subsec-Detection-Alg}.

Denote by $\Lambda_{i,0}$ the LMGF of $L_i(t)$ under hypothesis $H_0$:
\begin{equation*}
\label{eq-log-momentgen-fcn-non-id} \Lambda_{i,0}: \mathbb{R} \longrightarrow \left(-\infty,+\infty\right],\;\;\; \Lambda_{i,0}(\lambda)=\log \E \left[e^{\lambda L_i(1)}| H_0\right].
\end{equation*}
We assume finiteness of the LMGF's of all sensors. Assumption~2
is restated explicitly as Assumption~{B}.

%
%\begin{assumption}

\emph{Assumption B}
\label{ass_finite_LMGF-non-id}
For $i=1,\ldots N$, $\Lambda_{i,0}(\lambda)<+\infty$, $\forall \lambda \in {\mathbb R}$.
%\end{assumption}

The optimal centralized detector, with highest error exponent, is the likelihood ratio test with zero threshold $\gamma=0$~\cite{DemboZeitouni}, its error exponent is equal
to the Chernoff information of the vector of all sensors observations, and can be expressed
in terms of the LMGF's as:
\begin{eqnarray*}
C_{\mathrm{tot}}= \max_{\lambda \in [0,1]} - \sum_{i=1}^N \Lambda_{i,0}(\lambda) = -\sum_{i=1}^N \Lambda_{i,0}(\lambda^\bullet).
\end{eqnarray*}
Here, $\lambda^\bullet$ is the minimizer of $\sum_{i=1}^N \Lambda_{i,0}$ over $[0,1]$.
 We are now ready to state our results on the error exponent of the consensus+innovation detector for the case of non-identically distributed observations. (We continue to use
 $\alpha_i(k,\gamma)$, $\beta_i(k,\gamma)$, and $P_{\mathrm{e},i}(k,\gamma)$
  to denote the false alarm, miss, and Bayes error probabilities of distributed detector
  at sensor $i$.)
   \begin{theorem}
   \label{theorem-bounds-non-id-case}
Let Assumptions~A,~B and 3 hold, and let, in addition, $|\log r|>0.$ Consider the family of distributed detectors in~\eqref{eqn-running-cons-sensor-i} and~\eqref{eq-distributed-dec-rule} with thresholds $\gamma \in (\overline{\gamma}_0, \overline {\gamma}_1)$. Then,
at each sensor~$i$:
\begin{eqnarray}
\label{theorem-bound-on-alpha-non-id}
\liminf_{k\rightarrow \infty}-\frac{1}{k}\log \alpha_i(k,\gamma)
\geq B_0(\gamma)>0,\:\:
\liminf_{k\rightarrow \infty}-\frac{1}{k}\log \beta_i(k,\gamma)
\geq B_1(\gamma)>0,
\end{eqnarray}
where
\vspace{-4mm}
\begin{eqnarray}
\label{theorem-bound-on-alpha-non-id}
B_0(\gamma)
&=&\max_{ \phantom{-}\lambda \in [0,1]} N \lambda \gamma - \max \left\{\sum_{i=1}^N\Lambda_{i,0}(\lambda), \max_{i=1,\ldots,N} \Lambda_{i,0}(N \lambda)-|\log r|\right\}\\
\label{theorem-bound-on-beta-non-id}
B_1(\gamma)
&=& \max_{\lambda \in [-1,0]} N \lambda \gamma -\max \left\{\sum_{i=1}^N\Lambda_{i,1}(\lambda), \max_{i=1,\ldots,N} \Lambda_{i,1}(N \lambda)-|\log r|\right\}.
\end{eqnarray}
\end{theorem}
\begin{corollary}
\label{Corollary-main-result-non-id}
Let Assumptions A, B and 3 hold, and let, in addition, $|\log r|>0.$ Consider the family of distributed detectors in~\eqref{eqn-running-cons-sensor-i} and~\eqref{eq-distributed-dec-rule} with thresholds $\gamma \in (\overline{\gamma}_0, \overline {\gamma}_1)$.
Then:
\begin{enumerate}[(a)]
\item At each sensor $i$:
\begin{eqnarray}
\label{eqn-rate-generic}
\liminf_{k \rightarrow \infty} -\frac{1}{k} \log P_{\mathrm e,i}(k,\gamma) \geq \min\{B_0(\gamma),B_1(\gamma)\}>0,
\end{eqnarray}
and the lower bound in \eqref{eqn-rate-generic} is maximized for the point $\gamma^\star \in (\overline{\gamma}_0,\overline{\gamma}_1)$ at which $B_0(\gamma^\star)=B_1(\gamma^\star).$
%
%
%\begin{eqnarray}
%\label{eqn-rate-generic}
%\liminf_{k \rightarrow \infty} -\frac{1}{k} \log P_{\mathrm e,i}(k) \geq \left\{ \begin{array}{lll} \phantom{-} C_{\mathrm{tot}}, &
%\mathrm{if} \;\;\;  |\log r| \geq \mathrm{thr}\left(\Lambda_{1,0},\ldots,\Lambda_{N,0}\right)\\
%-\max \left\{ B_0, B_1 \right\}, & \mathrm {otherwise,} \end{array}, \right.
%\end{eqnarray}
%
%
\item Consider $\lambda^\bullet = \mathrm{arg\,min}_{\lambda \in [0,1]} \,\sum_{i=1}^N \Lambda_{i,0}(\lambda)$,
and let:
\begin{eqnarray}
\label{eqn-threshold-new}
\mathrm{thr}\left(\Lambda_{1,0},\ldots,\Lambda_{N,0}\right) =
\:\:\:\:\:\:\:\:\:\:
\:\:\:\:\:\:\:\:\:\:
\:\:\:\:\:\:\:\:\:\:
\:\:\:\:\:\:\:\:\:\:
\:\:\:\:\:\:\:\:\:\:
\:\:\:\:\:\:\:\:\:\:
\:\:\:\:\:\:\:\:\:\:
\:\:\:\:\:\:\:\:\:\:
\:\:\:\:\:\:\:\:\:\:
\:\:\:\:\:\:\:\:\:\:
\:\:\:\:\:\:\:\:\:\:
\:\:\:\:\:\:\:\:\:\:
 \\
\max \left\{\max_{i=1,\ldots,N} \Lambda_{i,0}(N \lambda^\bullet)-\sum_{i=1}^N \Lambda_{i,0}(\lambda^\bullet),
\max_{i=1,\ldots,N} \Lambda_{i,0}(1-N(1- \lambda^\bullet))-\sum_{i=1}^N \Lambda_{i,0}(\lambda^\bullet)
\right\}. \nonumber
\end{eqnarray}
Then, when $|\log r|\geq \mathrm{thr}\left(\Lambda_{1,0},\ldots,\Lambda_{N,0}\right) $, each sensor $i$ with the detector threshold set to $\gamma=0$, is asymptotically optimal:
 \[
 \lim_{k \rightarrow \infty}-\frac{1}{k} \log P_{\mathrm e,i}(k,0) = C_{\mathrm{tot}}.
 \]
%with $\lambda^\star = \mathrm{arg\,min}_{\lambda \in \mathbb R} \,\sum_{i=1}^N \Lambda_{i,0}(\lambda)$,
%and
%\begin{eqnarray}
%\label{bounds-B0}
%B_0&=& \min_{\lambda \in [0,1]} \max\left\{  \sum_{i=1}^N \Lambda_{i,0}(\lambda), \max_{i=1,\ldots,N} \Lambda_{i,0}(N \lambda)-|\log r| \right\}\\
%\label{bounds-B1}
%B_1&=& \min_{\lambda \in [0,1]} \max\left\{  \sum_{i=1}^N \Lambda_{i,0}(1-\lambda), \max_{i=1,\ldots,N} \Lambda_{i,0}(1-N \lambda)-|\log r|\right\}.
%\end{eqnarray}
%
%
%
%\item
%Moreover, if $|\log r| \geq \mathrm{thr}\left(\Lambda_{1,0},\ldots, \Lambda_{N,0}\right)$,
%the distributed detector~\eqref{eqn-running-cons-sensor-i}
%is asymptotically optimal at each sensor $i$, i.e.,
%\[
%\lim_{k \rightarrow \infty} -\frac{1}{k} \log P_{\mathrm e,i}(k) = C_{\mathrm{tot}}.
%\]
\end{enumerate}
\end{corollary}
Comparing Theorem \ref{Theorem-rate-bounds-for-alpha-and-beta} with Theorem \ref{theorem-bounds-non-id-case}, we can see that, under non-identically distributed observations, it is no longer possible to analytically characterize the lower bounds on the error exponents, $B_0(\gamma)$ and $B_1(\gamma)$. However, the objective functions (in the variable $\lambda$) in~\eqref{theorem-bound-on-alpha-non-id} and~\eqref{theorem-bound-on-beta-non-id} are concave (by convexity of the LMGF's) and the underlying optimization variable $\lambda$ is a scalar, and, thus, $B_0(\gamma)$ and $B_1(\gamma)$ can be efficiently found by a one dimensional numerical optimization procedure, e.g., a subgradient algorithm~\cite{Urruty}.

\begin{proof}[Proof of Theorem~\ref{theorem-bounds-non-id-case}] The proof of Theorem~\ref{theorem-bounds-non-id-case} mimics the proof of Theorem~\ref{Theorem-rate-bounds-for-alpha-and-beta}; we focus only on the steps that account for
different sensors' LMGF's. First, expression~\eqref{eq-Conditioning-on-Wl} that upper bounds the probability of false alarm $\alpha_i(k,\gamma)$ for the case of non-identically distributed observations becomes:
\begin{eqnarray}
\label{eq-Conditioning-on-Wl-non-id}
\E\left[e^{N\lambda  \sum_{t=1}^k\sum_{j=1}^N \Phi_{i,j}(k,t) L_j(t)}|H_0\right]&=&\E\left[\E\left[e^{N\lambda  \sum_{t=1}^k\sum_{j=1}^N \Phi_{i,j}(k,t) L_j(t)}|H_0,W(1),\ldots,W(k)\right]\right]\nonumber \\
&=& \E\left[e^{\sum_{t=1}^k\sum_{j=1}^N \Lambda_{j,0}\left(N\lambda \Phi_{i,j}(k,t)\right)}\right]. \nonumber
\end{eqnarray}
Next, we bound the sum in the exponent of the previous equation, conditioned on the event $\mathcal{A}_s$, for a fixed $s$ in $\{0,1\ldots,k\}$, deriving a counterpart to Lemma~\ref{lemma-Bounds-on-Lambda}.
\begin{lemma} Let Assumptions A, B, and 3 hold. Then,
\begin{enumerate}[(a)]
\item For any realization of $W(t)$, $t=1,2,...,k$:
  \[\sum_{j=1}^N\Lambda_{j,0}\left(N\lambda \Phi_{i,j}(k,t)\right)\leq  \max_{j=1,\ldots,N} \Lambda_{j,0}\left(N\lambda\right),\: \forall t=1, \ldots, k.\]
 \item Consider a fixed $s$ in $\{ 0, 1,...,k\}$. If the event
$\mathcal{A}_s$ occurred,
then, for $i = 1,...,N$:
%That is, conditioned on the event $\mathcal{A}_s$, we bound the exponent in~\eqref{eq-Conditioning-on-Wl-non-id} as follows:
\[\sum_{j=1}^N \Lambda_{j,0}\left(N\lambda \Phi_{i,j}(k,t)\right) \leq  \sum_{j=1}^N \max\left( \Lambda_{j,0}\left(\lambda-\epsilon N \sqrt {N}\lambda\right), \Lambda_{j,0}\left(\lambda+\epsilon N \sqrt {N}\lambda\right) \right),\: \forall t=1, \ldots, s.\]
\end{enumerate}
\end{lemma}
The remainder of the proof proceeds analogously to the proof of Theorem~\ref{Theorem-rate-bounds-for-alpha-and-beta}.
\end{proof}
 \section{Conclusion}
   \label{section-conclusion}
   We analyzed the large deviations performance (error exponent) of consensus+innovations distributed detection
   over random networks. The sensors' observations have generic (non-Gaussian) distribution, independent, not necessarily
   identical over space, and i.i.d. in time. Our results hold assuming that the log-moment
   generating functions of each sensor's log-likelihood ratio are finite. We showed that the distributed detector
   exhibits a phase transition behavior with respect to the network connectivity, measured by $|\log r|$, where $r$ is the (exponential) rate of convergence in probability of the product $W(k)W(k-1) \cdots W(1)$ to
 the consensus matrix $J:=(1/N)1 1^\top$. When $|\log r|$ is
   above the threshold, the distributed detector
   has the same error exponent as the optimal centralized detector. We further showed
   that the optimality threshold depends on the type of the distribution of the sensor observations.
   Numerical and analytical studies illustrated this dependence for Gaussian, Laplace, and binary distributions
   of the sensors' observations.
\appendix
\section{Appendix}
\label{section-appendix}
\subsection{Proof of finiteness of the log-moment generating function under \eqref{eqn-polynomial}-\eqref{eqn-logarithmic-2}}
We now show that Assumption~\ref{ass_finite_LMGF} holds, i.e.,
that $\Lambda_0(\cdot)$ is finite for any $\lambda \in \mathbb R$, if~\eqref{eqn-polynomial} and~\eqref{eqn-polynomial-2} hold.
 The other combinations for finiteness of $\Lambda_0(\cdot)$ when 1) either~\eqref{eqn-polynomial}~or~\eqref{eqn-logarithmic};
  \textit{and} 2) either~\eqref{eqn-polynomial-2} or~\eqref{eqn-logarithmic-2} hold
  %new to when~\eqref{eqn-polynomial}~and~\eqref{eqn-polynomial-2} hold,
   can  be shown similarly, and, hence, for brevity, we do not consider these cases.
 Assume $m>0$ (the case $m<0$ can be treated analogously), fix $\lambda \in {\mathbb R}$ and consider:
 \begin{eqnarray}
 \label{eqn-here}
 \Lambda_0(\lambda) = \log \int_{y=-\infty}^{+\infty} e^{\lambda \, \log \frac{f_n(y-m)}{f_n(y)}} f_n(y) d y,
 \end{eqnarray}
 where we use the fact that the density under $H_1$ is
 $f_1(y)=f_n(y-m)$, i.e., $f_1(\cdot)$ is the shifted density $f_n(\cdot)$ (of the noise) under $H_0$.
 With $f_n(y)=c e^{-g(y)}$, \eqref{eqn-here}
  is rewritten as:
  \begin{eqnarray*}
  e^{\Lambda_0(\lambda)} &=& c\, \int_{y=-\infty}^{+\infty}
   e^{\lambda \left[ -g(y-m)+g(y)  \right]} e^{-g(y)} dy
   =
   c\, \int_{y=-\infty}^{+\infty} e^{-g(y) \left[ 1 - \lambda \left( 1-\frac{g(y-m)}{g(y)} \right) \right] }dy\\
   &=& c\, \int_{y=-\infty}^{0} e^{-g(y) \left[ 1 - \lambda \left( 1-\frac{g(y-m)}{g(y)} \right) \right]} dy
   + c\, \int_{y=0}^{+\infty} e^{-g(y) \left[ 1 - \lambda \left( 1-\frac{g(y-m)}{g(y)} \right) \right] }dy. \nonumber
  \end{eqnarray*}
Now, by~\eqref{eqn-polynomial}, for any $\epsilon_1 \in (0,\infty)$,
there exists $M_1 \in (0,\infty)$, so that
\begin{equation*}
\left( (\rho_{+}) - \epsilon_1 \right) y^{\tau_{+}} \leq g(y) \leq  \left( (\rho_{+}) + \epsilon_1 \right) y^{\tau_{+}},\:\: \forall y \geq M_1.
\end{equation*}
Further, we have that:
\begin{equation}
\label{eqn-ova}
\left( (\rho_{+}) - \epsilon_1 \right) (y-m)^{\tau_{+}} \leq g(y-m) \leq  \left( (\rho_{+}) + \epsilon_1 \right) (y-m)^{\tau_{+}},\:\: \forall y \geq M_1+m.
\end{equation}
Also, for any $\epsilon_2 \in (0,\infty)$, there exists $M_2 \in (0,\infty)$,
such that:
\begin{equation}
\label{eqn-ona}
(1-\epsilon_2) (y-m)^{\tau_{+}} \leq y^{\tau_{+}} \leq  (1+\epsilon_2) (y-m)^{\tau_{+}},\:\: \forall y \geq M_2.
\end{equation}
Now, combining \eqref{eqn-ova} and \eqref{eqn-ona}, we obtain:
\begin{equation}
\label{eqn-ova2}
(1-\epsilon_2) \frac{(\rho_{+})-\epsilon_1}{(\rho_{+})+\epsilon_1} \leq \frac{g(y-m)}{g(y)} \leq (1+\epsilon_2) \frac{(\rho_{+})+\epsilon_1}{(\rho_{+})-\epsilon_1},\:\:\forall y \geq M_3:=\max\{M_1+m,M_2\}.
\end{equation}
To upper bound the integral
$\int_{y=0}^{+\infty} e^{-g(y) \left[ 1 - \lambda \left( 1-\frac{g(y-m)}{g(y)} \right) \right] }dy$,
 we note that, by \eqref{eqn-ova2}, we can choose $M_3$ large enough, so that:
 $
 \left|1-\frac{g(y-m)}{g(y)}\right| \leq \frac{\epsilon_3}{|\lambda|}, \:\:\ \forall y \geq M_3,
$
 for arbitrary $\epsilon_3\in(0,1)$.
 Thus, we have:
 \begin{eqnarray*}
 &\,& \int_{y=0}^{+\infty} e^{-g(y) \left[ 1 - \lambda \left( 1-\frac{g(y-m)}{g(y)} \right) \right] }dy
  =  \int_{y=0}^{M_3} e^{-g(y) \left[ 1 - \lambda \left( 1-\frac{g(y-m)}{g(y)} \right) \right] } dy
  +
  \int_{y=M_3}^{+\infty} e^{-g(y) \left[ 1 - \lambda \left( 1-\frac{g(y-m)}{g(y)} \right) \right] } dy\\
  &\leq&
  M_4 + \int_{y=M_3}^{+\infty} e^{-g(y) \left[ 1 - |\lambda| \frac{\epsilon_3}{|\lambda|} \right] }d y %\\
  \leq M_4 + \int_{y=M_3}^{+\infty} e^{{-(1-\epsilon_3)}g(y) } dy \leq M_4 + M_5 < \infty.
 \end{eqnarray*}
 Finiteness of the integral $\int_{y=-\infty}^{0} e^{-g(y) \left[ 1 - \lambda \left( 1-\frac{g(y-m)}{g(y)} \right) \right] }dy$,
  using equation \eqref{eqn-polynomial-2}, can be proved in an analogous way. As $\lambda \in \mathbb R$
   is arbitrary, we conclude that $\Lambda_0(\lambda) < +\infty$, $\forall \lambda \in \mathbb R$.
\bibliographystyle{IEEEtran}
\bibliography{IEEEabrv,Bibliography}

% Generated by IEEEtran.bst, version: 1.13 (2008/09/30)
\begin{thebibliography}{10}
\providecommand{\url}[1]{#1}
\csname url@samestyle\endcsname
\providecommand{\newblock}{\relax}
\providecommand{\bibinfo}[2]{#2}
\providecommand{\BIBentrySTDinterwordspacing}{\spaceskip=0pt\relax}
\providecommand{\BIBentryALTinterwordstretchfactor}{4}
\providecommand{\BIBentryALTinterwordspacing}{\spaceskip=\fontdimen2\font plus
\BIBentryALTinterwordstretchfactor\fontdimen3\font minus
  \fontdimen4\font\relax}
\providecommand{\BIBforeignlanguage}[2]{{%
\expandafter\ifx\csname l@#1\endcsname\relax
\typeout{** WARNING: IEEEtran.bst: No hyphenation pattern has been}%
\typeout{** loaded for the language `#1'. Using the pattern for}%
\typeout{** the default language instead.}%
\else
\language=\csname l@#1\endcsname
\fi
#2}}
\providecommand{\BIBdecl}{\relax}
\BIBdecl

\bibitem{GaussianDD}
D.~Bajovi\'{c}, D.~Jakoveti\'{c}, J.~Xavier, B.~Sinopoli, and J.~M.~F. Moura,
  ``Distributed detection via {G}aussian running consensus: Large deviations
  asymptotic analysis,'' \emph{IEEE Transactions on Signal Processing},
  vol.~59, no.~9, pp. 4381--4396, September 2011.

\bibitem{Moura-saddle-point}
S.~A. Aldosari and J.~M.~F. Moura, ``Detection in sensor networks: the
  saddlepoint approximation,'' \emph{IEEE Transactions on Signal Processing},
  vol.~55, no.~1, pp. 327--340, January 2007.

\bibitem{Varshney-I}
R.~Viswanatan and P.~R. Varshney, ``Decentralized detection with multiple
  sensors: Part {I}--fundamentals,'' \emph{Proc. IEEE}, vol.~85, no.~1, pp.
  54--63, January 1997.

\bibitem{Poor-II}
R.~S. Blum, S.~A. Kassam, and H.~V. Poor, ``Decentralized detection with
  multiple sensors: Part {II}--advanced topics,'' \emph{Proc. IEEE}, vol.~85,
  pp. 64--79, January 1997.

\bibitem{Veraavali-SPMag}
J.~F. Chamberland and V.~V. Veeravalli, ``Wireless sensors in distributed
  detection applications,'' \emph{IEEE Signal Processing Magazine}, vol.~24,
  no.~3, pp. 16--25, May 2007.

\bibitem{Veraavali}
J.~F. Chamberland and V.~Veeravalli, ``Decentralized dectection in sensor
  networks,'' \emph{IEEE Transactions on Signal Processing}, vol.~51, no.~2,
  pp. 407--416, February 2003.

\bibitem{Veraavali-LDP}
------, ``Asymptotic results for decentralized detection in power constrained
  wireless sensor networks,'' \emph{IEEE Journal on {S}elected {A}reas in
  {C}ommunications}, vol.~22, no.~6, pp. 1007--1015, August 2004.

\bibitem{sensor-selection}
D.~Bajovi\'{c}, B.~Sinopoli, and J.~Xavier, ``Sensor selection for event
  detection in wireless sensor networks,'' \emph{IEEE Transactions on Signal
  Processing}, vol.~59, no.~10, pp. 4938--4953, Oct. 2011.

\bibitem{Moura-detection-consensus}
S.~Kar, S.~A. Aldosari, and J.~M.~F. Moura, ``Topology for distributed
  inference on graphs,'' \emph{IEEE Transactions on Signal Processing},
  vol.~56, no.~6, pp. 2609–--2613, June 2008.

\bibitem{consensus-detection}
M.~Alanyali, S.~Venkatesh, O.~Savas, and S.~Aeron, ``Distributed detection in
  sensor networks with packet losses and finite capacity links,'' \emph{IEEE
  Transactions on Signal Processing}, vol.~54, no.~11, pp. 4118--4132, July
  2004.

\bibitem{aldosarimouramay06}
S.~A. Aldosari and J.~M.~F. Moura, ``Topology of sensor networks in distributed
  detection,'' in \emph{ICASSP'06, IEEE International Conference on Acoustics,
  Speech and Signal Processing}, vol.~5, Toulouse, France, May 2006, pp. 1061
  -- 1064.

\bibitem{MouraInference}
S.~Kar, S.~Aldosari, and J.~M.~F. Moura, ``Topology for distributed inference
  on graphs,'' \emph{IEEE Transactions on Signal Processing}, vol.~56, no.~6,
  pp. 2609--2613, June 2008.

\bibitem{Sayed-LMS}
C.~G. Lopes and A.~H. Sayed, ``Diffusion least-mean squares over adaptive
  networks: formulation and performance analysis,'' \emph{IEEE Transactions on
  Signal Processing}, vol.~56, no.~7, pp. 3122–--3136, July 2008.

\bibitem{SoummyaEst}
\BIBentryALTinterwordspacing
S.~Kar, J.~M.~F. Moura, and K.~Ramanan, ``Distributed parameter estimation in
  sensor networks: Nonlinear observation models and imperfect communication,''
  \emph{IEEE Transactions on Information Theory}, 2012, accepted for
  publication, 51 pages. [Online]. Available: \url{arXiv:0809.0009v1 [cs.MA]}
\BIBentrySTDinterwordspacing

\bibitem{Giannakis-LMS}
I.~D. Schizas, G.~Mateos, and G.~B. Giannakis, ``Distributed {LMS} for
  consensus-based in-network adaptive processing,'' \emph{IEEE Trans. on Signal
  Processing}, vol.~57, no.~6, pp. 2365–--2381, June 2009.

\bibitem{Giannakis-LMS-2}
G.~Mateos, I.~D. Schizas, and G.~B. Giannakis, ``Distributed recursive
  least-squares for consensus-based in-network adaptive estimation,''
  \emph{IEEE Trans. on Signal Processing}, vol.~57, no.~11, pp. 4583–--4588,
  November 2009.

\bibitem{Sayed-LMS-new}
F.~S. Cattivelli and A.~H. Sayed, ``Diffusion {LMS} strategies for distributed
  estimation,'' \emph{IEEE Transactions on Signal Processing}, vol.~58, no.~3,
  pp. 1035–--1048, March 2010.

\bibitem{Sayed-detection}
------, ``Distributed detection over adaptive networks based on diffusion
  estimation schemes,'' in \emph{Proc. IEEE SPAWC '09, 10th IEEE International
  Workshop on Signal Processing Advances in Wireless Communications}, Perugia,
  Italy, June 2009, pp. 61–--65.

\bibitem{Sayed-detection-2}
------, ``Diffusion {LMS}-based detection over adaptive networks,'' in
  \emph{Proc. Asilomar Conf. Signals, Systems and Computers}, Pacific Grove,
  CA, October 2009, pp. 171--175.

\bibitem{running-consensus-detection}
P.~Braca, S.~Marano, V.~Matta, and P.~Willet, ``Asymptotic optimality of
  running consensus in testing binary hypothesis,'' \emph{IEEE Transactions on
  Signal Processing}, vol.~58, no.~2, pp. 814--825, February 2010.

\bibitem{cattivellisayed-2011}
F.~Cattivelli and A.~Sayed, ``Distributed detection over adaptive networks
  using diffusion adaptation,'' \emph{IEEE Transactions on Signal Processing},
  vol.~59, no.~5, pp. 1917--1932, May 2011.

\bibitem{Soummya-Detection-Noise}
S.~Kar, R.~Tandon, H.~V. {P}oor, and S.~Cui, ``Distributed detection in noisy
  sensor networks,'' in \emph{Proc. ISIT 2011, Internetional Symposium on
  Information Theory}, Saint Petersburgh, Russia, August 2011, pp. 2856--2860.

\bibitem{stankovic-change-detection}
S.~S. Stankovi\'{c}, N.~Ili\'{c}, M.~S. Stankovi\'{c}, and K.~H. Johansson,
  ``Distributed change detection based on a consensus algorithm,'' \emph{IEEE
  Transactions on Signal Processing}, vol.~59, no.~12, pp. 5686--5697, December
  2011.

\bibitem{stankovic-conference}
------, ``Distributed change detection based on a randomized consensus
  algorithm,'' in \emph{ECCSC '10, 5th European Conference on Circuits and
  Systems for Communications}, Belgrade, Serbia, March 2010, pp. 51--54.

\bibitem{Kassam}
S.~A. {K}assam, \emph{Signal Detection in Non-Gaussian Noise}.\hskip 1em plus
  0.5em minus 0.4em\relax New York: Springer-Verlag, 1987.

\bibitem{allerton}
D.~Bajovi\'{c}, D.~Jakoveti\'{c}, J.~Xavier, B.~Sinopoli, and J.~M.~F. {M}oura,
  ``Distributed detection over time varying networks: large deviations
  analysis,'' in \emph{48th Allerton Conference on Communication, Control, and
  Computing}, Monticello, IL, Oct. 2010, pp. 302--309.

\bibitem{DusanNoise}
\BIBentryALTinterwordspacing
D.~Jakoveti\'{c}, J.~M.~F. Moura, and J.~Xavier, ``Distributed detection over
  noisy networks: large deviations analysis,'' August 2011, submitted for
  publication. [Online]. Available: \url{arXiv:1108.1410v1 [cs.IT]}
\BIBentrySTDinterwordspacing

\bibitem{Hollander}
F.~d.~Hollander, \emph{Large deviations}.\hskip 1em plus 0.5em minus
  0.4em\relax Fields {I}nstitute {M}onographs, American {M}athematical
  {S}ociety, 2000.

\bibitem{weinwright}
M.~J. Wainwright and M.~I. Jordan, ``Graphical models, exponential families,
  and variational inference,'' \emph{Foundations and Trends in Machine
  Learning}, vol.~1, no. 1--305, pp. 4938--4953, Oct. 2008.

\bibitem{DemboZeitouni}
A.~Dembo and O.~Zeitouni, \emph{Large {D}eviations {T}echniques and
  {A}pplications}.\hskip 1em plus 0.5em minus 0.4em\relax {B}oston, {MA}: Jones
  and Barlett, 1993.

\bibitem{rate-of-consensus}
D.~Bajovic, J.~Xavier, J.~M.~F. {M}oura, and B.~Sinopoli, ``Consensus and
  products of random stochastic matrices: exact rate for convergence in
  probability,'' \emph{submitted for publication}, February 2012, available at:
  http://arxiv.org/abs/1202.6389.

\bibitem{Karr}
A.~F. Karr, \emph{Probability}.\hskip 1em plus 0.5em minus 0.4em\relax {N}ew
  {Y}ork: Springer-Verlag, 1993.

\bibitem{BoydsBook}
S.~Boyd and L.~Vandenberghe, \emph{Convex Optimization}.\hskip 1em plus 0.5em
  minus 0.4em\relax {C}ambrige, {U}nited {K}ingdom: Cambridge University Press,
  2004.

\bibitem{SoummyaConferenceConsensus}
S.Kar and J.~M.~F. Moura, ``Distributed average consensus in sensor networks
  with random link failures,'' in \emph{ICASSP '07, IEEE International
  Conference on Acoustics, Speech and Signal Processing}, vol.~2, Pacific
  Grove, CA, April 2007, pp. {II}--1013--{II}--1016.

\bibitem{Urruty}
J.-B.~H. Urruty and C.~Lemarechal, \emph{Convex Analysis and Minimization
  Algorithms: Part 1: Fundamentals}, ser. Grundlehren der Mathematischen
  Wissenschaften 305, 306.\hskip 1em plus 0.5em minus 0.4em\relax {B}erlin,
  {G}ermany: Springer-Verlag, 1993.

\end{thebibliography}
\end{document}